\documentclass[lettersize,journal]{IEEEtran}
\usepackage{amsmath}
\usepackage{amsthm}
\usepackage{exscale,relsize}
\usepackage{cite}
\usepackage{graphicx}
\usepackage{cases}
\usepackage{epstopdf}
\usepackage{amsfonts,amsmath,amssymb}
\usepackage{graphicx}
\usepackage{bm}
\usepackage{bbm}
\usepackage{color}
\usepackage{enumerate}
\usepackage{hyperref}
\DeclareMathOperator*{\st}{s.t.}
\theoremstyle{definition}

\newtheorem{prop}{Proposition}

\newtheorem{Remark}{Remark}

\usepackage{algorithm,algorithmic}

\usepackage[belowskip=-5pt,aboveskip=0pt]{caption}
\allowdisplaybreaks[2]

\usepackage{float}
\begin{document}
%\bstctlcite{IEEEexample:BSTcontrol}

\title{{Task-Oriented Computation Offloading for Edge Inference: An Integrated Bayesian Optimization and Deep Reinforcement Learning Framework}

\author{Xian~Li,~\IEEEmembership{Senior Member,~IEEE}, Suzhi~Bi,~\IEEEmembership{Senior Member, IEEE}, and Ying-Jun Angela Zhang,~\IEEEmembership{Fellow, IEEE}}

\thanks{An earlier version of this paper has been accepted by the IEEE International Conference on Communications (ICC) 2026 \cite{liLAB2026}.
%	This work was supported in part by the National Natural Science Foundation of China under Project 62271325; in part by the Guangdong Basic and Applied Basic Research Foundation under Project 24050000527, 2022A1515010973, 2414050002238, and 23201910240002621; and in part by the Shenzhen Science and Technology Program under Project 20220810142637001 and JCYJ20210324093011030. The associate editor coordinating the review of this paper and approving it for publication was Pan Li. (Corresponding author: Suzhi Bi.)\
	
Xian Li and Suzhi Bi are with The State Key Laboratory of Radio Frequency Heterogeneous Integration (Shenzhen University), and also with the College of Electronics and Information Engineering, Shenzhen University, Shenzhen, China 518060.(email: \{xianli,~bsz\}@szu.edu.cn).

Ying-Jun Angela Zhang is with the Department of Information Engineering, The Chinese University of Hong Kong, Hong Kong (email: yjzhang@ie.cuhk.edu.hk).}}
\maketitle

%\vspace{1000pt}artificial intelligence (AI) model (e.g., For example, in a wireless sensing application where the AI model gradually refines the intermediate feature, an SD with deep-fading channel should have a late splitting point which computes more DNN layers at local and offloads a small intermediate output to the ES for remaining edge inference. 
% The work was supported in part by the National Natural Science Foundation of China (Project 61871271), the Guangdong Province Pearl River Scholar Funding Scheme 2018, the Key Project of Department of Education of Guangdong Province (No. 2020ZDZX3050), and the Foundation of Shenzhen City (Project JCYJ20170818101824392, JCYJ20190808120415286). Corresponding authors: Suzhi Bi.
\begin{abstract}
Edge intelligence (EI) allows resource-constrained edge devices (EDs) to offload computation-intensive AI tasks (e.g., visual object detection) to edge servers (ESs) for fast execution. However, transmitting high-volume raw task data (e.g., 4K video) over bandwidth-limited wireless networks incurs significant latency. While EDs can reduce transmission latency by degrading data before transmission (e.g., reducing resolution from 4K to 720p or 480p), it often deteriorates inference accuracy, creating a critical accuracy-latency tradeoff. The difficulty in balancing this tradeoff stems from the absence of closed-form models capturing content-dependent accuracy-latency relationships. Besides, under bandwidth sharing constraints, the discrete degradation decisions among the EDs demonstrate inherent combinatorial complexity. Mathematically, it requires solving a challenging \textit{black-box} mixed-integer nonlinear programming (MINLP). To address this problem, we propose LAB, a novel learning framework that seamlessly integrates deep reinforcement learning (DRL) and Bayesian optimization (BO). Specifically, LAB employs: (a) a DNN-based actor that maps input system state to degradation actions, directly addressing the combinatorial complexity of the MINLP; and (b) a BO-based critic with an explicit model built from fitting a Gaussian process surrogate with historical observations, enabling model-based evaluation of degradation actions. For each selected action, optimal bandwidth allocation is then efficiently derived via convex optimization. Numerical evaluations on real-world self-driving datasets demonstrate that LAB achieves near-optimal accuracy-latency tradeoff, exhibiting only 1.22\% accuracy degradation and 0.07s added latency compared to exhaustive search. Notably, it outperforms conventional DRL with 3.29\% higher accuracy and 42.60\% lower latency, demonstrating its advantageous performance in handling black-box optimization problems. The complete source code for LAB is available at \url{https://github.com/Ethan-Xian-Li/LAB}.
\end{abstract}


\IEEEpeerreviewmaketitle

\section{Introduction}
\begin{figure}
	\centering
	\includegraphics[scale=0.445]{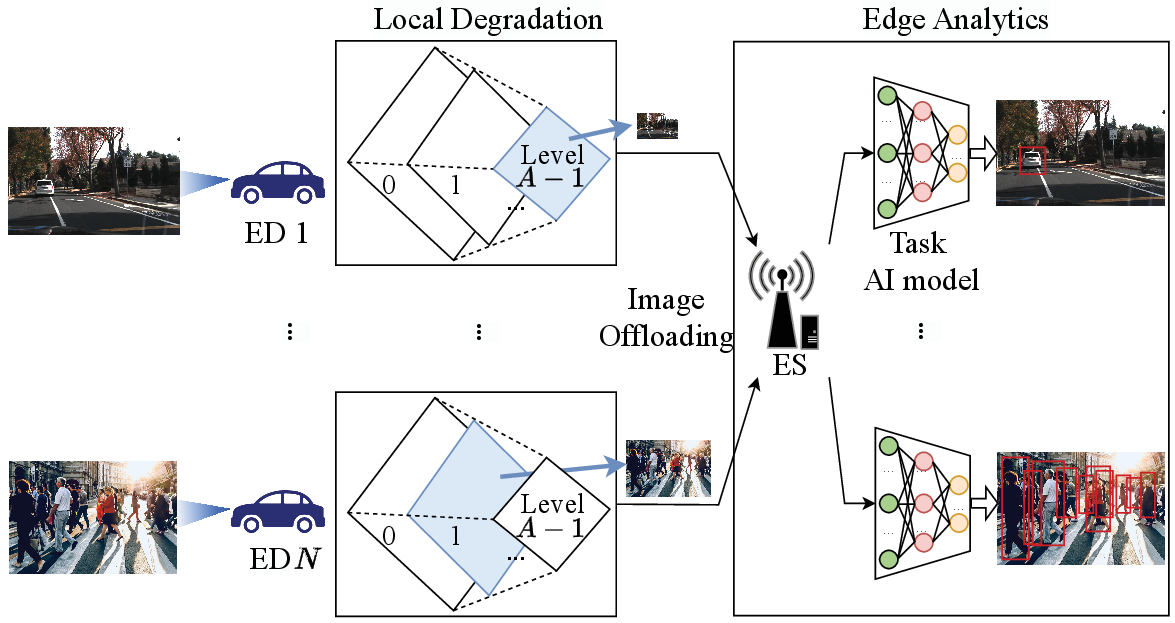}
	\captionsetup{font=footnotesize}
	\caption{An example of task-oriented computation offloading in an EI system for visual object detection. }
	\label{sys_mod}
\end{figure}

{\color{black}{Edge inference (EI) enables resource-constrained edge devices (EDs) to execute computation-intensive AI tasks (e.g., pedestrian detection and cooperative perception for autonomous vehicles \cite{baiCollaborative2025,yangEdge2021, liResourceEfficient2026}) through task-oriented computation offloading.}} This paradigm offloads computations from EDs to proximal edge servers (ESs) to optimize task-specific performance metrics like detection accuracy. However, transmitting high-fidelity raw data over bandwidth-limited wireless links incurs significant latency. To address this, EDs may perform data degradation before transmission. As Fig. \ref{sys_mod} illustrates for visual object detection, EDs can select from multiple degradation levels and downsample high-resolution (e.g., 4K) video streams to lower-resolution (e.g., 720p or 480p). While this degradation reduces transmission latency, it sacrifices inference accuracy, creating an inherent accuracy-latency tradeoff. To achieve an optimal tradeoff under time-varying channel and content, an adaptive offloading scheme is required.

Designing adaptive offloading for edge-AI inference tasks presents fundamental challenges. Crucially, AI inference quality is highly sensitive to the task content. Identical degradation operations (e.g., the same degradation level) may yield divergent accuracy outcomes for different input contents. For example, in visual object detection, sparse scenes (e.g., a single vehicle on an empty street) often tolerate aggressive degradation as critical features remain discernible, whereas complex scenes (e.g., crowded pedestrian crossings) demand high-fidelity inputs to maintain recognition accuracy. Unlike conventional data-centric metrics (e.g., computation rate and energy consumption) which have explicit formulations, task-specific metrics like accuracy exhibit strong content dependency, making it difficult to build closed-form analytical models. Instead, it creates \textit{black-box} optimization problems that require computationally heavy experimental evaluations. Besides, practical EI deployments typically involve multiple EDs competing for limited bandwidth resources. Optimizing system-wide performance thus requires joint coordination of both discrete degradation choice per ED and continuous bandwidth allocation across competing devices. The tight coupling of hybrid decision variables, compounded by the absence of a closed-form objective function, poses a challenging black-box mixed-integer nonlinear programming (MINLP) problem. 

Due to the black-box nature, conventional MINLP solvers (e.g., branch-and-bound \cite{morrisonBranchandbound2016} and metaheuristics like particle swarm optimization (PSO) \cite{shamiParticle2022}) are fundamentally unsuited for dynamic EI environments. Their trial-based nature requires on-device execution of multiple candidate actions to select the best one. For instance, PSO requires EDs to transmit numerous degraded versions to the ES to find the best action from its particle swarm candidates. Such trial-and-execute operations consume substantial bandwidth and computational resources, rendering them fundamentally incompatible with latency-sensitive EI systems. Alternatively, Bayesian optimization (BO) \cite{brochuTutorial2010} overcomes this limitation by eliminating the on-device execution of candidate actions. Specifically, it builds probabilistic surrogate models of black-box objective function from historical observations, and designs acqusition functions to select proper actions \cite{wangRecent2023}. However, standard BO becomes computationally prohibitive for MINLP problems at practical network scales, where the global optimization of the acquisition function is impeded by the combinatorial action selection. While deep reinforcement learning (DRL) \cite{arulkumaranDeep2017,lillicrapContinuous2019} enables rapid decision-making for adaptive offloading, its model-free nature requires extensive environmental interactions to train its policy. This results in poor sample efficiency and premature convergence to ineffective polices, hindering its application in dynamic EI environments.

In this paper, we propose a novel learning framework for task-oriented computation offloading that combines the real-time decision-making strength of DRL with the sample efficiency of BO. As shown in Fig. \ref{sys_mod}, we consider multiple EDs performing visual object detection tasks, and aim to maximize the long-term average detection accuracy while minimizing inference latency. The main contributions are summarized as follows:

\begin{itemize}
	\item \textit{Task-oriented Joint Computation Offloading and Resource Allocation}: We formulate a joint computation offloading and resource allocation problem to optimize task-specific objectives including detection accuracy and latency. The major challenge lies in the lack of an explicit closed-form function that maps optimization variables to the objective. Besides, the combinatorial nature of discrete degradation control variables across EDs prohibits real-time adaptation to fast-varying wireless environments.
	\item \textit{Integrated DRL and BO Framework}: To solve the target problem, we propose LAB, a novel learning framework that seamlessly integrates the real-time decision-making of DRL and the model-based evaluation of BO. Specifically, LAB features a \textit{DNN-based actor} that generates a compact set of high-potential degradation actions given an input system state, thereby overcoming the combinatorial complexity. These candidate degradation actions are then evaluated by a \textit{BO-based critic}, which fits a probabilistic surrogate model on historical execution data and analytically ranks candidate actions via closed-form acqusition functions. For the selected highest-ranked action, the optimal bandwidth allocation is derived via convex optimization. 
	
	\item \textit{Exploration-Efficient Actor Design}: We design an exploration-efficient DNN-based actor that generates a diverse yet high-quality set of candidate degradation actions by sampling near high-preference regions of the DNN's output space. The actor dynamically adjusts the candidate set size according to the distance between empirically optimal actions and the DNN's predicted preference. Larger distances trigger expansion of the exploration range to discover potentially better actions, and vice versa. Such a self-adjusting mechanism maintains an efficient exploration-exploitation trade-off, enabling robust adaptation to dynamic visual environments without compromising computational efficiency.
	\item \textit{Real-World Dataset Validation}: We evaluate LAB on a real-world self-driving dataset for object detection. The LAB framework demonstrates: (a) Near-ideal accuracy with merely 1.22\% performance deficit and only 0.07s additional end-to-end inference latency versus exhaustive search; (b) 3.29\% higher accuracy with 42.60\% lower latency than conventional DRL approaches; (c) 32.96\% faster inference than standard BO at equivalent accuracy, with computation overhead reduced to less than $1/2000$ of the original. 
\end{itemize}

\section{Related Works}
Computation offloading is a well-established technique in mobile edge computing (MEC), enabling resource-constrained EDs to leverage a nearby ES for executing demanding tasks. Conventional MEC offloading strategies primarily follow two paradigms: binary offloading \cite{Huang2019,Nguyen2019a,duQoEGuaranteed2026,biLyapunovGuided2021a}, where tasks are executed entirely locally or remotely, and partial offloading \cite{liEnergyEfficient2022,duProfit2025,wangMultiantenna2019,gaoTask2023,liOptimal2024}, which partitions computational workloads between EDs and the ES. Despite effectively optimizing system-level metrics (e.g., computation rate, energy consumption) using explicit objectives, these data-centric approaches treat task data as generic payloads, ignoring its critical role as content carrier for AI performance. 
As a result, they disregard how processing operations (e.g, resolution degradation) alter input content and consequently degrade task-specific performance such as inference accuracy. Conventional offloading strategies \cite{Huang2019,Nguyen2019a,duQoEGuaranteed2026,biLyapunovGuided2021a,liEnergyEfficient2022,duProfit2025,liOptimal2024,wangMultiantenna2019,gaoTask2023} therefore fail to optimize the accuracy-latency tradeoff in EI systems.

Emerging approaches mitigate this challenge by explicitly characterizing action-to-performance relationships through empirical modeling. For example, Yang et al. \cite{yangEfficient2025} fitted an empirical model between synthesis accuracy and denoising steps in diffusion models using DreamBooth simulations \cite{ruizDreamBooth2023}, enabling the joint optimization of user scheduling and denoising partitions to balance latency and accuracy. Similarly, Liu et al. \cite{liuResource2023} leveraged pre-deployment profiling of DNNs (e.g., ResNet) to establish monotonic layer-to-accuracy relationships in multi-user EI systems with early exiting. Building on this foundation, they jointly optimized user scheduling and exit-point selection to maximize system throughput under per-user latency and accuracy constraints. While effective in static environments, these empirical fitting-based approaches \cite{yangEfficient2025,liuResource2023} require substantial task-specific data collection to build models, and their static mappings cannot adapt to dynamic environments during operation.

To address environmental dynamics, researchers have developed online learning optimizers that continuously refine decisions through real-time feedback. Representative approaches includes DRL frameworks \cite{wangEdge2023} and BO-based strategies \cite{tangBayesian2024,galanopoulosAutoML2021a}. For instance, Wang et al. \cite{wangEdge2023} proposes a real-time DRL solution for joint frame-resolution adaptation and bandwidth allocation in edge video analytics to balance latency and detection accuracy. While standard DRL enables online policy updates, its model-free treatment of action-to-performance relationships demands intensive environmental interactions during training. This ultimately undermines its resilience to rapid changes and raises the risk of non-convergence in dynamic environments. In contrast, BO-based approaches offer a complementary solution by constructing sample-efficient surrogate models. For example, Galanopoulos et al. \cite{galanopoulosAutoML2021a} pioneered an AutoML framework for video analytics configuration using BO principles, avoiding costly trial executions through surrogate model-based evaluation. Similarly, Tang et al. \cite{tangBayesian2024} applied BO to the identical joint adaptation problem addressed by \cite{wangEdge2023}, achieving comparable accuracy-latency performance while demonstrating superior temporal stability and user fairness. Both BO strategies, however, face computational intractability when optimizing hybrid discrete-continuous action spaces at large network scales. Such intractability stems from the combinatorial explosion during acquisition function maximization, incurring prohibitive decision latency incompatible with real-time EI processing demands.

Given the advantages of both approaches, Gong et al. \cite{gongBayesian2023} stand as one of the earliest and still rare attempts at combining BO and DRL, demonstrated in the context of UAV trajectory optimization. Crucially, their method treats BO and the DRL actor as parallel action generators, producing competing actions subsequently evaluated and selected by a conventional DNN-based critic. This structure maintains a black-box evaluation process and consequently inherits core DRL limitations in adaptation speed and sample efficiency. In contrast, our work introduces a fundamentally redesigned DRL-BO framework. Within this architecture, a DNN-based actor efficiently identifies candidate solutions while a dedicated BO-based critic performs rapid evaluation using closed-form acquisition functions. This integrated design takes advantage of complementary strengths of BO and DRL to efficiently tackle black-box optimization problems in dynamic EI environments.

\section{System Model and Problem Formulation}
\subsection{System Model}
As shown in Fig. \ref{sys_mod}, we consider an EI system consisting of an ES supporting $N\geq1$ EDs. The system operates over $T>0$ sequential time slots, with each ED performing an individual object detection task per slot. Due to the constrained computational resources at the EDs, the task AI model supporting multi-resolution inputs resides on the ES. Each ED collects RGB images from its surrounding environment and relies on the ES for edge processing. During each time slot $t$, the system executes a sequential workflow comprising local degradation at EDs, image offloading to the ES, and edge analytics at the ES.   

1) \textit{Local Degradation}: At the beginning of time slot $t$, each ED$_n$ ($1\leq n\leq N$) captures a raw image at native resolution $\bm{r}_{n}^{\rm o}=(\iota_n^{\rm w}, \iota_n^{\rm h})$, where $\iota_n^{\rm w}$ and $\iota_n^{\rm h}$ denote the pixel dimensions in width and height, respectively. To alleviate communication overhead, ED$_n$ optionally applies local degradation techniques to reduce transmission data volume. In this paper, we implement local degradation through Gaussian pyramid downsampling \cite{burtLaplacian1983}, which utilizes dyadic spatial scaling for resolution reduction. Specifically, ED$_n$ selects an integer degradation level $a_{t,n}$ from a finite discrete set $\mathcal{A}=\{0,1,\cdots,A-1\}$, where $A$ denotes the maximum available degradation levels. Then, the offloading resolution of ED$_n$ during time slot $t$ is
\begin{equation}\label{eq_offres}
	\bm{r}_{t,n} = \left(\frac{\iota_n^{\rm w}}{2^{a_{t,n}}},\frac{\iota_n^{\rm h}}{2^{a_{t,n}}}\right),
\end{equation}
where $a_{t,n}=0$ preserves native resolution, and each integer increment of $a_{t,n}$ halves spatial dimensions. The computational latency $\tau_{t,n}^{\rm d}$ required for this resolution adaptation is given by:
\begin{equation}\label{eq_latencyd}
	\tau_{t,n}^{\rm d}=f_{\rm d}(\bm{r}_n^{\rm o}, a_{t,n};\psi_n^{\rm d}),
\end{equation}
where $f_{\rm d}(\cdot)$ is a device-specific function and $\psi_n^{\rm d}$ characterizes the computational efficiency parameter of the ED$_n$ \cite{wangEdge2023}.

2) \textit{Image Offloading}: Following local degradation, each ED transmits its processed image to the ES for inference. To avoid co-channel interference among EDs, we implement frequency division multiple access (FDMA) for uplink transmissions, where the total system bandwidth $W$ is dynamically allocated to EDs. Let $b_{t,n}W$ represent the bandwidth allocated to the ED$_n$ during time slot $t$, where $b_{t,n}\geq 0$ is the fractional allocation variable satisfying that 
\begin{equation}\label{eq_bandconst}
	\sum_{n=1}^{N}b_{t,n}\leq 1, \forall t=1,\cdots,T.
\end{equation}
The transmission rate for ED$_n$ is 
\begin{equation}\label{eq_comrate}
	R_{t,n} = b_{t,n}W\log_2\left(1+\frac{p_nh_{t,n}}{b_{t,n}W\delta^2}\right),
\end{equation}
where $p_n$ is the fixed transmit power at the ED$_n$. $\delta^2$ is the power spectrum density of additive white Guassian noise (AWGN) at the ES. $h_{t,n}$ denotes the channel gain between the ED$_n$ and the ES during time slot $t$. We consider a block fading scenario where $h_{t,n}$ remains constant during time slot $t$ but varies independently across time slots. 

Given the offloading resolution $\bm{r}_{t,n}$ in \eqref{eq_offres}, the offloading data size $d_{t,n}$ in bits for an RGB image with 24 bits/pixel (i.e., 3 color channels/pixel $\times$ 8 bits/channel) is computed as
\begin{equation}\label{eq_offdatasize}
	d_{t,n}=\frac{\iota_n^{\rm w}}{2^{a_{t,n}}}\times \frac{\iota_n^{\rm h}}{2^{a_{t,n}}} \times 24.
\end{equation}
The corresponding offloading time is
\begin{equation}\label{eq_latencyo}
	\tau_{t,n}^{\rm o} = \frac{d_{t,n}}{R_{t,n}}.
\end{equation}

3) \textit{Edge Analytics}: Upon receiving the preprocessed image from ED$_n$, the ES performs object detection using a pre-deployed AI model (e.g., YOLO \cite{redmonYou2016}). The model outputs a set of bounding boxes $\hat{\mathcal{B}}_{t,n}$ defining object locations with confidence values, alongside the edge computational latency $\tau_{t,n}^{\rm c}$. Since the degradation level $a_{t,n}$ determines the input resolution $\bm{r}_{t,n}$, the model outputs $\hat{\mathcal{B}}_{t,n}$ and $\tau_{t,n}^{\rm c}$ critically depend on $a_{t,n}$.

In this paper, we quantify the edge analytics performance through two primary metrics: detection accuracy and edge computational latency. The detection accuracy is quantified by comparing the predicted bounding boxes $\hat{\mathcal{B}}_{t,n}$ against ground truth $\mathcal{B}^{\rm{GT}}_{t,n}$ \cite{wangEdge2023}. Specifically, for each ground truth bounding box $\beta \in \mathcal{B}^{\rm{GT}}_{t,n}$, we compute its maximum Intersection over Union (IoU) with predicted boxes $\hat{\beta} \in \hat{\mathcal{B}}_{t,n}$ as:
\begin{equation}
	\Gamma(\beta,\hat{\mathcal{B}}_{t,n}) = \max_{\hat{\beta} \in \hat{\mathcal{B}}_{t,n}} \dfrac{|\beta \cap \hat{\beta}|}{|\beta\cup\hat{\beta}|},
\end{equation}
where $|\cdot|$ represents set cardinality. The detection accuracy $c_{t,n}$ is then defined as the proportion of ground truth boxes exceeding an IoU threshold $\gamma_{\rm th}$:
\begin{equation}
	c_{t,n} = \frac{1}{|\mathcal{B}^{\rm{GT}}_{t,n}|} \sum_{\beta \in \mathcal{B}^{\rm{GT}}_{t,n}}\mathbb{I}\left(\Gamma(\beta,\hat{\mathcal{B}}_{t,n}) \geq \gamma_{\rm th} \right),
\end{equation}
where $\mathbb{I}(\cdot)$ denotes the indicator function.

The edge computational latency $\tau_{t,n}^{\rm c}$ depends primarily on the original resolution $\bm{r}_n^{\rm o}$ and the degradation level $a_{t,n}$, but also reflects the computational efficiency $\psi$ at the ES. We model this relationship through a device-specific function $f_c(\cdot)$ as:
\begin{equation}\label{eq_latencyc}
	\tau_{t,n}^{\rm c} = f_{\rm c}(\bm{r}_n^{\rm o}, a_{t,n}; \psi).
\end{equation}

After completing inference tasks (e.g., generating bounding boxes), the ES returns the results directly to the originating EDs. Given the substantially smaller data volume of these inference outputs compared to original and degraded input images, we omit feedback latency in our analysis.

%\begin{figure*}
%	\centering
%	\includegraphics[scale=0.55]{Figure//System_model.eps}
%	\captionsetup{font=footnotesize}
%	\caption{The considered wireless sensing system: (a) the system model; (b) the AI model; and (c) the task graph of AI model.}
%	\label{sys_mod}
%\end{figure*}

\subsection{Problem Formulation}\label{sec2}
In this paper, we address an online offloading control problem to simultaneously maximize the long-term average detection accuracy of EDs and minimize end-to-end inference latency. We define the utility for ED$_n$ in time slot $t$ as
\begin{equation}\label{eq_utility}
	u_{t,n} = c_{t,n} - w_n\tau_{t,n},
\end{equation}
where $w_n\geq0$ is an accuracy-latency tradeoff weighting factor and
\begin{equation}
	\tau_{t,n}=\tau_{t,n}^{\rm d} + \tau_{t,n}^{\rm o} + \tau_{t,n}^{\rm c}
\end{equation}
is the end-to-end inference latency. Our objective is to maximize the long-term average utility through joint optimization of degradation levels and bandwidth allocation:
\begin{subequations}\label{prob1}
	\begin{align}
	\underset{\substack{\bm{a}_t,\bm{b}_t,\forall t}} \max~&\frac{1}{T}\sum_{t=0}^{T-1}\sum^{N}_{n=1} u_{t,n} \label{prob1_obj} \\ 
	\st	
	~~& \smaller{\sum}_{n=1}^{N}b_{t,n}\leq 1, \forall t, \label{prob1_const1}\\
	~~& a_{t,n} \in \mathcal{A}, b_{t,n} \geq 0, \forall n, t. \label{prob1_const2}
	\end{align}
\end{subequations}
Here, $\bm{a}_t=\{a_{t,n},\forall n\}$ denotes the discrete degradation actions, and $\bm{b}_t=\{b_{t,n},\forall n\}$ represents the bandwidth allocation ratios, respectively. \eqref{prob1_const1} ensures feasible bandwidth allocation, while \eqref{prob1_const2} enforces variable boundaries.

During real-time inference, the detection accuracy $c_{t,n}$ remains unavailable due to the absence of ground truth $\mathcal{B}^{\rm{GT}}_{t,n}$. To address this limitation, we leverage the confidence values associated with predicted bounding boxes $\hat{\mathcal{B}}_{t,n}$ from the object detection model. The cumulative confidence metric $\alpha_{t,n}$, defined as the sum of confidence values across all predicted bounding boxes $\hat{\mathcal{B}}_{t,n}$, provides a real-time accuracy proxy strongly correlated with actual detection performance \cite{galanopoulosAutoML2021a}. This enables us to reformulate the following alternative optimization problem {\footnote{{\color{black}While it admits an equivalent per-slot decomposition, the long-term formulation in \eqref{prob2} represents a general starting point. It provides a consistent framework that can naturally accommodate future extensions involving inter-slot dependencies, e.g., device energy budgets and queue stability.}}}:
\begin{subequations}\label{prob2}
	\begin{align}
		\underset{\substack{\bm{a}_t,\bm{b}_t},\forall t} \max~&\frac{1}{T}\sum_{t=0}^{T-1} U_t(\bm{h}_t, \bm{a}_t,\bm{b}_t) \label{prob2_obj} \\ 
		\st	
		~~& \eqref{prob1_const1}, \eqref{prob1_const2},  \label{prob2_const1}
	\end{align}
\end{subequations}
where $\bm{h}_{t}=\{h_{t,n}, \forall n\}$ denotes the channel state of time slot $t$. $U_t(\bm{h}_t, \bm{a}_t,\bm{b}_t) = \sum^{N}_{n=1} \tilde{u}_{t,n}$ is the total utility of all EDs, with per-ED utility defined as
\begin{equation}\label{eq_utility_alpha}
	\tilde{u}_{t,n} = \alpha_{t,n} - w_n\tau_{t,n}.
\end{equation}
At the start of time slot $t$, the ES has perfect knowledge of $\bm{h}_t$ and employs centralized control to determine $\bm{a}_t$ and $\bm{b}_t$. Despite the temporal decoupling in \eqref{prob2} allows independent optimization per slot $t$, solving each per-slot subproblem (i.e., \eqref{prob2} for fixed $t$) faces three interconnected challenges:
\begin{itemize}
	\item \textit{Lack of Content Awareness}: The confidence value $\alpha_{t,n}$ exhibits significant content-dependent variations, yet the dynamic task content (i.e., image content) governing $\alpha_{t,n}$ remains unknown at the ES during optimization. While each ED could evaluate all degradation options to determine the optimal choice, this requires pre-transmitting every possible degraded image version to the ES, which is impractical as it induces prohibitive latency through excessive transmissions.  
	\item \textit{Absence of Explicit Utility Model}: Even with perfect task content knowledge, the reformulated utility function $U_t$ lacks an analytical model. Specifically, cumulative confidence $\alpha_{t,n}$ depends on complex objective detection model behaviors that defy close-form formulation. Meanwhile, degradation latency $\tau_{t,n}^{\rm d}$ and edge computational latency $\tau_{t,n}^{\rm c}$ depend on the implicit hardware efficiency of individual ED and ES, respectively. These empirical relationships must be learned through system measurements, preventing direct gradient-based optimization or analytical solutions.
	\item \textit{Combinatorial Explosion}: The tight coupling between discrete degradation actions $\bm{a}_t$ and continuous bandwidth allocation $\bm{b}_t$ renders \eqref{prob2} an NP-hard MINLP. A conventional MINLP approach would fix discrete actions and optimize continuous variables using established solvers (e.g., CVX \cite{boydConvex2004}), then select the optimal discrete action through enumeration. However, exhaustive search over the $\mathcal{A}^N$ combinations incurs combinatorial explosion. This approach becomes computationally prohibitive for large edge deployments, as complexity grows exponentially with the number of EDs $N$.  	 
\end{itemize}

\begin{figure}
	\centering
	\includegraphics[scale=0.57]{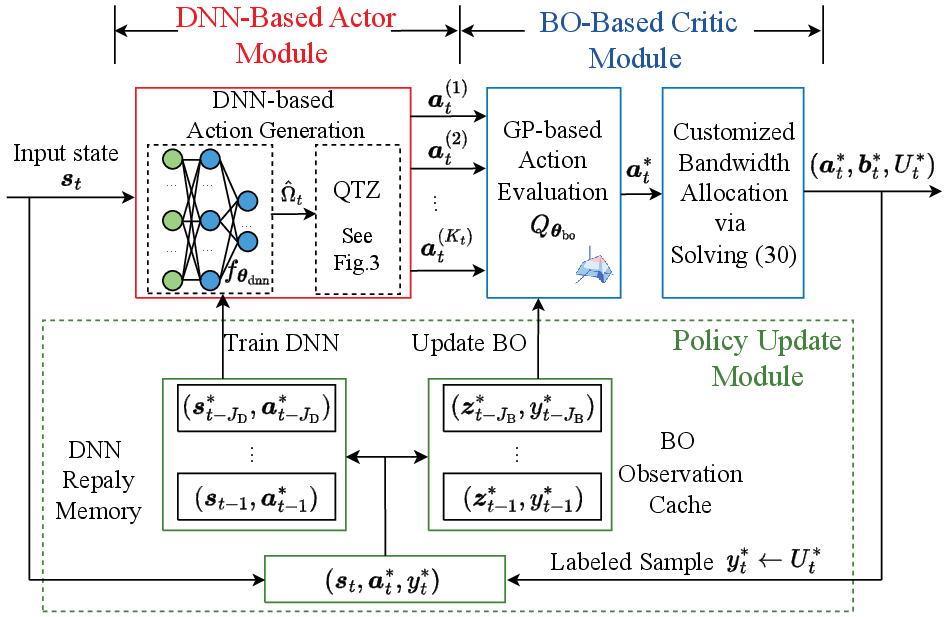}
	\captionsetup{font=footnotesize}
	\caption{The schematics of the proposed LAB framework.}
	\label{fig_LABframework}
\end{figure}

\subsection{Solution Framework Overview}
In this paper, we propose an integrated deep reinforcement \textbf{L}earning \textbf{a}nd  \textbf{B}ayesian optimization framework (LAB) to solve \eqref{prob2}. As shown in Fig. \ref{fig_LABframework}, LAB comprises three interconnected modules: a DNN-based actor parameterized by $\bm{\theta}_{\rm dnn}$, a BO-based critic modeled by $\bm{\theta}_{\rm bo}$, and a policy update module that enables continuous adaptation. During time slot $t$, the actor module first processes the input state $\bm{s}_t$ (including channel states and history experiences, formally defined in Section \ref{sec4}) to generate $K_t$ candidate degradation actions $\{\bm{a}_t^{(k)}\}_{k=1}^{K_t}$, where $\bm{a}_t^{(i)}$ is the $i$-th candidate. These candidates undergo evaluation by the critic module using Gaussian process (GP) modeling, which characterizes utility $U_t$ as a probabilistic function of degradation actions. Building on this probabilistic characterization, the critic computes an acquisition function (detailed in Section \ref{sec3}) to select the optimal candidate $\bm{a}_t^\ast$ that maximizes the acquisition value. Following this selection, it then determines the corresponding bandwidth allocation $\bm{b}_t^\ast$ by solving subproblem, thereby establishing the joint action-resource solution $(\bm{a}_t^\ast, \bm{b}_t^\ast)$. Concluding the operational cycle, the policy update module periodically refines both $\bm{\theta}_{\rm dnn}$ and $\bm{\theta}_{\rm bo}$ using accumulated historical execution data.

As an seamless integration of BO and DRL, the proposed LAB framework delivers three key innovations that directly address the challenges outlined below \eqref{eq_utility_alpha}:
\begin{itemize}
	\item \textit{GP Surrogate for Content Uncertainty Mitigation}: The BO-based critic employs GP to construct a probabilistic utility surrogate from historical observations, which captures the relationship between degradation actions and their expected utility across diverse task contents. By optimizing the acquisition function over this surrogate, LAB selects optimal degradation actions $\bm{a}_t$ without requiring real-time content knowledge or pre-transmission of degraded task images.
	
	\item \textit{Co-Adaptive Optimization of Implicit Utility}: LAB coordinates a DNN-based actor and a BO-based critic for online utility optimization. The BO-based critic pinpoints high-value actions from history observations to guide the actor's policy updates. In turn, the DNN-based actor generates candidate actions that focus the GP model on high-potential regions, enhancing BO efficiency. This co-adaptation enables derivative-free optimization of the utility without an explicit formulation.
	
	\item \textit{Exponential-to-Linear Complexity Reduction}: LAB generates a focused set of $K_t \ll |\mathcal{A}|^N$ candidate degradation actions via the DNN actor. This enables the BO critic to evaluate only $K_t$ options via acquisition computation, avoiding the intractable combinatorial action selection. Consequently, the framework solves the bandwidth allocation subproblem \eqref{prob_bandalloc} only once for the optimal candidate, collectively achieving linear-scale complexity.
	
\end{itemize}

In the following sections, we first present the critic module (Section \ref{sec3}) to establish the action evaluation framework. Building upon this foundation, we then introduce the actor module (Section \ref{sec4}) to generate candidates specifically for this evaluation process. Finally, the policy update module (Section \ref{sec5}) completes the operational cycle of LAB by refining the parameters of both actor and critic.

\section{BO-based Critic Module}\label{sec3}
To address the dual challenges of content uncertainty and utility implicitness in \eqref{prob2}, the BO-based critic module employs GP surrogates to evaluate candidate actions. Crucially, we observe that for a fixed degradation action $\bm{a}_t$, the utility $U_t$ becomes an explicit function of $\bm{b}_t$. As shown in Fig. \ref{fig_LABframework}, this structural insight enables us to decompose the critic module into two sequential stages: 1) GP-based action selection via acquisition function maximization, and 2) customized bandwidth optimization for the selected degradation action. We introduce these two stages in Sections \ref{sec3-1} and \ref{sec3-2}, respectively.

\subsection{GP-based Action Selection}\label{sec3-1}
Let $\mathcal{C}_t = \{\bm{a}_t^{(k)}\}_{k=1}^{K_t}$ denote the candidate action set produced by the DNN-based actor in time slot $t$. For each candidate degradation action $\bm{a}_t\in \mathcal{C}_t$, there exists an optimal bandwidth allocation $\bm{b}_t^\ast$ that maximizes the utility $U_t$ (as formally established in Section \ref{sec3-2}). We therefore define the maximum achievable utility for action $\bm{a}_t$ as $U_t(\bm{h}_t, \bm{a}_t) \triangleq \max_{\bm{b}_t} U_t(\bm{h}_t, \bm{a}_t, \bm{b}_t)$. This allows us to formulate the degradation action selection problem for time slot $t$ as:
\begin{subequations}\label{prob3}
	\begin{align}
		\underset{\substack{\bm{a}_t}} \max~& U_t(\bm{h}_t, \bm{a}_t) \label{prob3_obj} \\ 
		\st	
		~~& \bm{a}_t \in \mathcal{C}_t. \label{prob3_const1}
	\end{align}
\end{subequations} 
Despite this formulation, \eqref{prob3} faces two persistent challenges identified in Section \ref{sec2}: $U_t(\bm{h}_t, \bm{a}_t)$ remains content-dependent and model-free, while exhaustive evaluation of $U_t$ would require executing all $K_t$ candidates, incurring prohibitive latency. To address these limitations, we develop a GP surrogate that enables efficient model-based evaluation for any given $\bm{a}_t$ without on-device execution. The GP-based action selection stage contains two sequential operations: probabilistic modeling of the utility function $U_t(\bm{h}_t, \bm{a}_t)$ through GP regression, and acquisition function optimization leveraging the GP posterior distribution.

\subsubsection{Probabilistic GP Modeling}
To start with, we construct an augmented input vector $\bm{z}_t=\left(\bm{h}_t, \bm{a}_t, t\right)$ and express the utility compactly as $U(\bm{z}_t):= U_t(\bm{h}_t, \bm{a}_t)$. Here, the time index $t$ is explicitly included in $\bm{z}_t$ to account for temporal dynamics introduced by time-varying channel conditions and task content. The observed utility $y_t$ constitutes a noisy measurement of the true utility $U(\bm{z}_t)$:
\begin{equation}
	y_t = U(\bm{z}_t) + \epsilon_t,
\end{equation}
where $\epsilon_t\sim\mathcal{N}(0, \delta_\epsilon^2)$ is an independent and identically distributed (i.i.d) Gaussian noise with variance $\delta_\epsilon^2$. Crucially, $\delta_\epsilon$ is a learnable parameter that quantifies the uncertainty arising from content variations and channel randomness. 

Let $\mathcal{D}_{t-1}^{\rm B}=\{\left(\bm{z}_i, y_i\right)\}_{i=t-J_{\rm B}}^{t-1}$ denote all observations accumulated from the last $J_{\rm B}$ time slots, where $J_{\rm B}\in\{1, \cdots, t-1\}$. Our goal is to infer the utility value $U(\bm{z}_t)$ at the start of a new time slot $t$ using historical data $\mathcal{D}_{t-1}^{\rm B}$. To achieve this, we use a zero-mean GP prior to model the utility function:
\begin{equation}\label{eq_gp_prior}
	U(\bm{z}) \sim \mathcal{GP}(0, \kappa(\bm{z},\bm{z}^\prime)),
\end{equation}
where the covariance kernel function $\kappa(\cdot, \cdot)$ quantifies the similarity between any pair of inputs. Consequently, the joint prior distribution of utility evaluations $\bm{U}_{t-1}:=\left[U(\bm{z}_{t-J_{\rm B}}), \cdots, U(\bm{z}_{t-1})\right]^\top$ at inputs $\bm{Z}_{t-1}:=\left[\bm{z}_{t-J_{\rm B}}, \cdots, \bm{z}_{t-1}\right]^\top$ follows a multivariate Gaussian:
\begin{equation}\label{eq_prior}
	q\left(\bm{U}_{t-1}|\bm{Z}_{t-1}\right) = \mathcal{GP}(\bm{0}_{J_{\rm B}\times 1}, \bm{K}_{t-1}),
\end{equation}
where $\bm{0}_{J_{\rm B}\times 1}$ is a $J_{\rm B}\times1$ zero vector and $\bm{K}_{t-1}$ is a $J_{\rm B}\times J_{\rm B}$ covariance matrix with elements $[\bm{K}_{t-1}]_{j,j^\prime}=\kappa(\bm{z}_{t-j},\bm{z}_{t-j^\prime})$ for $j, j^\prime=1, \cdots, J_{\rm B}$. For the observed outputs $\bm{Y}_{t-1}:=[y_{t-J_{\rm B}}, \cdots, y_{t-1}]^\top$, the likelihood function given $\bm{U}_{t-1}$ and $\bm{Z}_{t-1}$ is
\begin{equation}\label{eq_likelihood}
	q\left(\bm{Y}_{t-1} \mid \bm{U}_{t-1}, \bm{Z}_{t-1}\right) = \mathcal{GP}(\bm{U}_{t-1}, \delta_\epsilon^2\bm{I}_{J_{\rm B}\times J_{\rm B}}),
\end{equation}
where $\bm{I}_{J_{\rm B}\times J_{\rm B}}$ denotes a $J_{\rm B}\times J_{\rm B}$ identity matrix. Combining this likelihood with the prior in \eqref{eq_prior} via Bayes' theorem, we obtain the posterior distribution for a new input $\bm{z}_t$ as \cite{brochuTutorial2010}
\begin{equation}\label{eq_postdist}
	q\left(U(\bm{z}_t)|\mathcal{D}_{t-1}^{\rm B}\right) = \mathcal{N}\left(\mu_{t-1}(\bm{z}_t), \sigma_{t-1}^2(\bm{z}_t)\right).
\end{equation}
The posterior mean $\mu_{t-1}$ and variance $\sigma_{t-1}^2$ are given in closed form by
\begin{equation}
	\mu_{t-1}(\bm{z}_t)=\bm{\kappa}_{t-1}^{\top}\left(\bm{K}_{t-1} + \delta_\epsilon^2\bm{I}_{J_{\rm B}\times J_{\rm B}}\right)^{-1}\bm{Y}_{t-1},
\end{equation}
and
\begin{equation}
	\!\!\sigma_{t-1}^2(\bm{z}_t)\! =\! \kappa(\bm{z}_t, \bm{z}_t)\! - \!\bm{\kappa}_{t-1}^{\top} \left(\bm{K}_{t-1} \!+ \!\delta_\epsilon^2\bm{I}_{J_{\rm B}\times J_{\rm B}}\right)^{-1}\!\bm{\kappa}_{t-1},
\end{equation}
respectively, where $\bm{\kappa}_{t-1}:=[\kappa(\bm{z}_{t-J_{\rm B}}, \bm{z}_t), \cdots, \kappa(\bm{z}_{t-1}, \bm{z}_t)]^\top$ is the covariance vector between $\bm{z}_t$ and the historical inputs.

The effectiveness of the above GP model critically depends on the kernel function $\kappa(\cdot,\cdot)$, which encodes prior information about the structure of utility function $U(\bm{z_t})$. To capture the complex interactions across the heterogeneous components of $\bm{z}_t = (\bm{h}_t, \bm{a}_t, t)$, we design the following composite kernel
\begin{equation}\label{eq_kernel}
	\begin{split}
		\kappa(\bm{z}_i, \bm{z}_{i^\prime}) = \kappa_{\rho}^{\rm{tmp}}(i, {i^\prime})&  \cdot \left[\kappa_{\bm{\theta}_{\rm h}}^{\rm{rbf}}(\bm{h}_i, \bm{h}_{i^\prime}) + \kappa_{\upsilon_{\rm a}}^{\rm{cat}}(\bm{a}_i, \bm{a}_{i^\prime}) \right.\\ 
		&\left.+ \kappa_{\bm{\theta}_{\rm h}}^{\rm{rbf}}(\bm{h}_i, \bm{h}_{i^\prime}) \cdot \kappa_{\upsilon_{\rm a}}^{\rm{cat}}(\bm{a}_i, \bm{a}_{i^\prime}) \right],
	\end{split}
\end{equation}
where $i, i^{\prime}\in\{t-J_{\rm B}, \cdots, t-1\}$ are time indices. The constituent kernels $\kappa_{\bm{\theta}_{\rm h}}^{\rm{rbf}}$, $\kappa_{\upsilon_{\rm a}}^{\rm{cat}}$, and $\kappa_{\rho}^{\rm{tmp}}$ are defined as follows:
%\begin{itemize}
%	\item $\kappa_{\bm{\theta}_{\rm h}}^{\rm{rbf}}(\bm{h}_i,\!\bm{h}_{i^{\prime}})\!\!=\!\!\upsilon_{\rm h} \cdot \exp\left(-\frac{1}{2}(\bm{h}_i \!- \!\bm{h}_{i^{\prime}})^\top {\rm diag}(\bm{\ell})^{-2}(\bm{h}_i\! -\! \bm{h}_{i^{\prime}})\right)$ is an RBF kernel parameterized by $\bm{\theta}_{\rm h}=\{\upsilon_{\rm h}, \bm{\ell}\}$, where $\upsilon_{\rm h}$ is the output scale and $\bm{\ell}$ is the length-scale vector, measuring similarity in continuous channel states. 
%	\item $\kappa_{\upsilon_{\rm a}}^{\rm{cat}}(\bm{a}_i, \bm{a}_{i^{\prime}}) = \frac{\upsilon_{\rm a}}{N}\sum_{n=1}^{N}\mathbb{I}(a_{i,n} = a_{i^{\prime},n})$ is a categorical kernel that compares discrete degradation actions, with $\upsilon_{\rm a}$ controlling the similarity strength;
%	\item $\kappa_{\rho}^{\rm{tmp}}(i, i^{\prime}) = (1-\rho)^{|i-i^{\prime}|/2}$ is a temporal kernel that models time-varying correlations, where $\rho \in (0,1)$ controls the decay rate of temporal dependence.
%\end{itemize}
\begin{itemize}
	\item Channel State Kernel $\kappa_{\bm{\theta}_{\rm h}}^{\rm{rbf}}$: To capture the similarity in continuous channel states $\bm{h}_t$, we utilize an automatic relevance determination (ARD) radial basis function (RBF) kernel:
	\begin{equation}
		\!\!\!	\kappa_{\bm{\theta}_{\rm h}}^{\rm{rbf}}(\bm{h}_i, \bm{h}_{i^{\prime}})\!=\upsilon_{\rm h}\cdot\!\exp\!\!\left(\!\!-\frac{1}{2}\left(\bm{h}_i\!-\!\bm{h}_{i^{\prime}}\right)^\top \!\!{\rm diag}(\bm{\ell})^{-2}\!\left(\bm{h}_i\!-\!\bm{h}_{i^{\prime}}\right)\!\right),
	\end{equation}
	where $\bm{\theta}_{\rm h}=\{\upsilon_{\rm h}, \bm{\ell}\}$. $\upsilon_{\rm h}$ is an output-scale parameter weighting channel state similarity. The diagonal matrix $\rm{diag}(\bm{\ell})$, formed from the length-scale vector $\bm{\ell} = [\ell_1, \cdots, \ell_N]^\top$, captures the relevance of each dimension in $\bm{h}_t$ for predicting $U_t$.
	\item Degradation Action Kernel $\kappa_{\upsilon_{\rm a}}^{\rm{cat}}$,: To measure the similarity between discrete degradation actions $\bm{a}_t$, we implement a categorical kernel \cite{yanBayesian2024}:
	\begin{equation}
		\kappa_{\upsilon_{\rm a}}^{\rm{cat}}(\bm{a}_i, \bm{a}_{i^{\prime}})=\frac{\upsilon_{\rm a}}{N}\sum_{n=1}^{N}\mathbb{I}\left(a_{i,n}=a_{{i^{\prime},n}}\right),
	\end{equation}
	where $\upsilon_{\rm a}$ is an output-scale parameter weighting categorical similarity.
	\item Temporal Dynamics Kernel $\kappa_{\rho}^{\rm{tmp}}$: To model time-varying correlations arising from evolving channel conditions and content characteristics, we employ a temporal kernel \cite{yanBayesian2024}:
	\begin{equation}
		\kappa_{\rho}^{\rm{tmp}}(i, i^{\prime})=(1-\rho)^{\frac{|i-i^{\prime}|}{2}},
	\end{equation}
	where $\rho\in(0,1)$ denotes the decay rate parameter controlling correlation persistence over time. This formulation ensures more recent observations exert stronger influence on predictions and especially suitable for non-stationary system behavior.
\end{itemize}

\subsubsection{Acquisition Optimization}Building upon the GP model above, we now define the acquisition function that guides the selection of degradation actions $\bm{a}_t$ at time slot $t$. We adopt the upper confidence bound (UCB) criterion \cite{wangRecent2023} to balance exploration of uncertain actions against exploitation of known high-utility actions. Given the posterior distribution in \eqref{eq_postdist}, the UCB acquisition function of input $\bm{z}_t$ is formulated as 
\begin{equation}\label{eq_acqfunc}
	Q_{\bm{\theta}_{\rm{bo}}}(\bm{z}_t) = \mu_{t-1}(\bm{z}_t) + \zeta \cdot \sigma_{t-1}(\bm{z}_t),
\end{equation}
where $\bm{\theta}_{\rm{bo}}=\{\bm{\theta}_{\rm h}, \upsilon_{\rm a}, \rho, \delta_\epsilon\}$ denotes the parameters of the BO-based critic. $\zeta > 0$ is an exploration coefficient that controls the trade-off between exploration and exploitation. The first term $\mu_{t-1}(\bm{z}_t)$ encourages exploitation by favoring actions expected to yield high utility, while the second term $\zeta \sigma_{t-1}(\bm{z}_t)$ promotes exploration of actions with high predictive uncertainty.

At the start of time slot $t$, given channel state information $\bm{h}_t$ at the ES, we optimize $\bm{a}_t$ by solving
\begin{equation}\label{prob_acqmax}
	\bm{a}_t^* = \underset{\bm{a}_t \in \mathcal{C}_t}{\arg\max} \ Q_{\bm{\theta}_{\rm{bo}}}\left( \bm{h}_t, \bm{a}_t, t \right).
\end{equation}
Given the finite candidate set $\mathcal{C}_t = \{{\bm{a}_t^{(i)}}\}_{i=1}^{K_t}$ where $K_t \ll |\mathcal{A}|^N$, we exhaustively evaluate $Q_{\bm{\theta}_{\rm{bo}}}$ over all $\bm{a}_t \in \mathcal{C}_t$ to obtain the optimal degradation action $\bm{a}_t^\ast$. This selected action is subsequently deployed for image offloading, with the optimal bandwidth allocation $\bm{b}_t^\ast$ to be determined in the next subsection.

%and the resulting utility observation $y_t$ is incorporated into $\mathcal{D}{t-1}$ for subsequent inference. This closed-loop process iteratively refines the utility model while adapting to temporal variations in channel conditions and content characteristics.

\subsection{Customized Optimal Bandwidth Allocation under Fixed Degradation Action}\label{sec3-2}
Given a selected degradation action $\bm{a}_t^\ast$, the utility $U_t$ in \eqref{prob2_obj} depends on bandwidth allocation $\bm{b}_t$ only through the offloading time $\tau_{t,n}^{\rm o}$, as $\alpha_{t,n}$, $\tau_{t,n}^{\rm d}$, and $\tau_{t,n}^{\rm c}$ are invariant to $\bm{b}_t$. This motivates introducing $\tau_{t,n}^{\rm o}$'s as auxiliary variables constrained by:
\begin{equation}\label{eq_auxi_const}
	\frac{d_{t,n}}{\tau_{t,n}^{\rm o}} \leq b_{t,n}W\log_2\left(1+\frac{p_nh_{t,n}}{b_{t,n}W\delta^2}\right), \tau_{t,n}^{\rm o}\geq 0, \forall n.
\end{equation}  
Consequently, the per-slot subproblem of \eqref{prob2} reduces to the following convex bandwidth allocation problem:
{\color{black}{\begin{subequations}\label{prob_bandalloc}
	\begin{align}
		\underset{\substack{\bm{b}_t}, \bm{\tau}_t^{\rm o}} \min~& \sum_{n=1}^N w_n \tau_{t,n}^{\rm o}\\
		\st
		~~& \frac{d_{t,n}}{\tau_{t,n}^{\rm o}} \leq b_{t,n}W\log_2\left(1+\frac{p_nh_{t,n}}{b_{t,n}W\delta^2}\right), \label{prob_bandalloc_const1}\\
		~~& \mathsmaller\sum_{n=1}^{N}b_{t,n}\leq 1, \label{eq_bandconst_perslot}\\
		~~& b_{t,n} \geq 0,  \tau_{t,n}^{\rm o}\geq 0, \forall n, \label{prob_bandalloc_const2}
	\end{align}
\end{subequations}}}
where $\bm{\tau}_t^{\rm o}=\{\tau_{t,n}^{\rm o}, \forall n\}$. While standard convex optimization methods (e.g., interior point method \cite{boydConvex2004}) efficiently solve \eqref{prob_bandalloc}, we derive a semi-closed-form solution using Lagrangian duality to characterize optimal resource allocation. {\color{black}{Let $\bm{\phi}=\{\phi_n\geq 0, \forall n\}$ and $\eta \geq 0$ denote Lagrangian multipliers for constraints \eqref{prob_bandalloc_const1} and \eqref{eq_bandconst_perslot}, respectively.}} The partial augmented Lagrangian of \eqref{prob_bandalloc} is
\begin{equation}\label{eq_lagran}
	\begin{split}
		&\bm{\mathcal{L}}(\bm{b}_t, \bm{\tau}_t^{\rm o}, \bm{\phi}, \eta)\\
		&= {\sum}_{n=1}^N w_n \tau_{t,n}^{\rm o} +\eta\left[{\sum}_{n=1}^N b_{t,n} - 1\right] \\
		&+ {\sum}_{n=1}^N\phi_n \left[\frac{d_{t,n}}{\tau_{t,n}^{\rm o}} - b_{t,n}W\log_2\left(1+\frac{p_nh_{t,n}}{b_{t,n}W\delta^2}\right)\right].
	\end{split}
\end{equation}
The optimum of \eqref{prob_bandalloc} is obtained through alternately solving the following bi-level optimization: 
\begin{enumerate}
	\item Obtain the dual function $\mathcal{K}(\bm{\phi}, \eta)$ by solving:
	\begin{subequations}\label{Dual_Func}
		\begin{align}
			\mathcal{K}(\bm{\phi}, \eta)=\underset{\bm{b}_t, \bm{\tau}_t^{\rm o}} \min~&\bm{\mathcal{L}}(\bm{b}_t, \bm{\tau}_t^{\rm o}, \bm{\phi}, \eta)\\
			~\st
			~~&b_{t,n}, \tau_{t,n}^{\rm o}\geq 0, \forall n.
		\end{align}
	\end{subequations}
	\item Solve the dual problem:
	\begin{equation}\label{Dual_Prob}
		\underset{\substack{\bm{\phi}, \eta}} \max~\mathcal{K}(\bm{\phi}, \eta),~~\st~~\bm{\phi}, \eta \geq 0.
	\end{equation}
\end{enumerate}
Let $\{\eta^\ast, \phi_n^\ast,\forall n\}$ denote the optimal dual variables. The optimal solution to \eqref{prob_bandalloc} is then characterized by the following proposition.
\begin{prop}\label{prop1}
	The optimal solution to \eqref{prob_bandalloc}, denoted as $\{b_{t,n}^{\ast}, \tau_{t,n}^{\rm o\ast}\}$,  has the following closed-form structure:
	\begin{align}
		&\tau_{t,n}^{\rm o\ast} = \sqrt{\frac{\phi_n^\ast d_{t,n}}{w_n}},\label{eq_opt_tau}\\
		&b_{t,n}^\ast = \frac{-p_nh_{t,n}}{W\delta^2\left[1+\left(\mathcal{W}\left(-\frac{1}{\exp\left(\frac{\eta^\ast \ln2}{\phi_n^\ast W}+1\right)}\right)\right)^{-1}\right]}, \label{eq_opt_b}
	\end{align}
	where $\mathcal{W}(\cdot)$ is the Lambert-$\mathcal{W}$ function and $\exp(\cdot)$ is the exponential function.
\end{prop}
\begin{proof}
	By taking the derivative of $\mathcal{L}$ in \eqref{eq_lagran} with respect to $\tau_{t,n}^{\rm o}$ and $b_{t,n}$, respectively, we have
	\begin{align}
		&\frac{\partial \bm{\mathcal{L}}}{\partial \tau_{t,n}^{\rm o}} = w_n - \frac{\phi_n d_{t,n}}{\left(\tau_{t,n}^{\rm o}\right)^2}, \label{proof_prop1_eq1}\\
		&\frac{\partial \bm{\mathcal{L}}}{\partial b_{t,n}}\! =\! -\frac{\phi_n W}{\ln 2}\!\left[\ln \!\left(1\!\!+\!\!\frac{p_nh_{t,n}}{b_{t,n}W\sigma^2}\right)\!-\!\frac{p_nh_{t,n}}{b_{t,n}W\sigma^2 \!+\! p_nh_{t,n}}\right]\!\! + \!\!\eta.\label{proof_prop1_eq2}
	\end{align}
	By setting \eqref{proof_prop1_eq1} to 0, we obtain
	\begin{align}
		&\tau_{t,n}^{\rm o\ast} = \sqrt{\frac{\phi_n^\ast d_{t,n}}{w_n}}.
	\end{align}
	By setting \eqref{proof_prop1_eq2} to 0, we have 
	\begin{equation}\label{proof_prop1_eq3}
		\!\!	\left(1\!+\!\frac{p_nh_{t,n}}{b_{t,n}W\sigma^2}\right)\exp\left(\frac{1}{1\!+\!\frac{p_nh_{t,n}}{b_{t,n}W\sigma^2}}\right)\! =\! \exp\left(1\!+\!\frac{\eta \ln2}{\phi_nW}\right).
	\end{equation}
	By solving \eqref{proof_prop1_eq3}, we obtain that 
	\begin{equation}
		b_{t,n}^\ast = \frac{-p_nh_{t,n}}{W\delta^2\left[1+\left(\mathcal{W}\left(-\frac{1}{\exp\left(\frac{\eta^\ast \ln2}{\phi_n^\ast W}+1\right)}\right)\right)^{-1}\right]},
	\end{equation}
	where $\mathcal{W}(\cdot)$ is the Lambert-$\mathcal{W}$ function.
\end{proof}

Proposition \ref{prop1} yields two key insights: First, $\phi_n^\ast>0$ necessarily holds because $\phi_n^\ast=0$ implies $\tau_{t,n}^{\rm o\ast}=0$, leading to unrealizable infinite data rates. Second, $\eta^\ast>0$ must obtain since $\eta^\ast=0$ causes $b_{t,n}^\ast\rightarrow \infty$, violating the bandwidth constraint \eqref{eq_bandconst_perslot}. These conditions ensure equality holds in both \eqref{prob_bandalloc_const1} and \eqref{eq_bandconst_perslot} at optimality, indicating that full utilization of bandwidth resource minimizes offloading time. The Lambert-$\mathcal{W}$ function has two real branches: the principal branch $\mathcal{W}_0(z) \in (-1,0)$ and the secondary branch $\mathcal{W}_{-1}(z) \in (-\infty,-1)$. Given $\phi_n^\ast>0$ and $\eta^\ast>0$, the argument $-\frac{1}{\exp\left(\frac{\eta^\ast \ln2}{\phi_n^\ast W}+1\right)} \in (-\frac{1}{e},0)$. Within this interval, $\mathcal{W}_{-1}(\cdot)$ yields negative values. To guarantee a non-negative $b_{t,n}^\ast$, we take $\mathcal{W}_0(\cdot)$ as the Lambert-$\mathcal{W}$ function in \eqref{eq_opt_b}.

\section{DNN-based Actor Module}\label{sec4}
The BO-based critic in Section \ref{sec3} provides probabilistic action evaluation through GP modeling with acquisition function optimization. While theoretically capable of evaluating all $\mathcal{A}^N$ actions, the exponential growth of the action space with $N$ makes exhaustive evaluation computationally prohibitive for real-time edge systems. Critically, even in computationally tractable cases (e.g., small $N$ where exhaustive search is feasible), the action maximizing the acquisition function would not guarantee optimality for the true utility $U_t$ in \eqref{prob2_obj}, as the acquisition function (e.g., UCB in \eqref{eq_acqfunc}) is a surrogate designed for exploration-exploitation balance rather than direct utility optimization. These dual limitations motivate our DNN-based actor module, which generates a targeted candidate set $\mathcal{C}_t$ through adaptive history distillation. This approach maintains computational feasibility while focusing critic evaluation on high-potential actions, thereby enhancing decision quality beyond standalone acquisition optimization. As shown in Fig. \ref{fig_LABframework}, the actor module decomposes into two sequential stages: 1) DNN-based approximation and 2) customized quantization (QTZ) for candidate generation. We detail these two stages in Sections \ref{sec4-1} and \ref{sec4-2}, respectively.

\subsection{DNN-based Action Approximation}\label{sec4-1}
We begin by formalizing the input state $\bm{s}_t$ of the DNN. At the start of time slot $t$, we define the system observation as $\bm{o}_t=\left(\bm{h}_{t}, \bm{\alpha}_{t-1}, \bm{\tau}_{t-1}, \bm{a}_{t-1}, \bm{b}_{t-1}, U_{t-1}\right)$, where $\bm{\alpha}_{t-1}=\{\alpha_{t-1,n}, \forall n\}$ and $\bm{\tau}_{t-1}=\{\tau_{t-1,n}, \forall n\}$ denote the confidence and the end-to-end inference latency of time slot $t-1$, respectively. The input state aggregates the most recent $l$ observations, forming the history sequence $\bm{s}_t = (\bm{o}_{t-l+1}, \dots, \bm{o}_t)$. For initialization at $t=1$, we set $\bm{\alpha}_{0} = \bm{0}_{N\times1}$, $\bm{\tau}_{0} = \bm{0}_{N\times1}$, $\bm{a}_{0} = \bm{0}_{N\times1}$, $\bm{b}_{0} = \bm{0}_{N\times1}$, and $U_{0} = 0$, where $\bm{0}_{N\times1}$ denotes an $N$-dimensional zero vector.

To enhance the exploration capacity of the DNN, we employ one-hot encoding for the integer degradation action. For each user $n$, the degradation level $a_{t,n}$ is represented as an $A$-dimensional binary vector $\bm{\omega}_{t,n} = [\omega_{t,n,0}, \cdots, \omega_{t,n,A-1}]$, where each element $\omega_{t,n,i}, \forall i\in\mathcal{A}$, is defined by the indicator function:
\begin{subnumcases}{\omega_{t,n,i}=\mathbb{I}(a_{t,n} = i)=}
	1, & $a_{t,n} = i$,\\
	0, &\text{otherwise}.
\end{subnumcases}
The degradation action $\bm{a}_t = \{a_{t,n}\}_{n=1}^N$ is consequently encoded as an $NA\times1$ binary vector $\bm{\Omega}_t = [\bm{\omega}_{t,1}, \cdots, \bm{\omega}_{t,N}]^\top$. The DNN, parameterized by $\theta_{\rm dnn}$ and denoted $f_{\theta_{\rm dnn}}$, approximates a continuous mapping from the input state $\bm{s}_t$ to a relaxed preference
\begin{equation}
	\hat{\bm{\Omega}}_t = f_{\theta_{\rm dnn}}(\bm{s}_t).
\end{equation}
Here, $\hat{\bm{\Omega}}_t$ maintains identical $NA$-dimensional concatenated structure as $\bm{\Omega}_t$, but with real-valued elements $\hat{\omega}_{t,n,i} \in [0,1]$ representing preference scores for each $n \in \{1,\dots,N\}$ and $i \in \mathcal{A}$.

\subsection{Customized Quantization for Candidates Generation}\label{sec4-2}
Building upon the continuous approximation $\hat{\bm{\Omega}}_t$ generated by the DNN in Section \ref{sec4-1}, this subsection develops a customized QTZ mechanism to produce feasible discrete candidate actions $\{\bm{a}_t^{(k)}\}_{k=1}^{K_t}$. We denote $\hat{\bm{\omega}}_{t,n}=[\hat{\omega}_{t,n,0},\cdots, \hat{\omega}_{t,n,A-1}]$ and $\hat{\bm{\Omega}}_t=[\hat{\bm{\omega}}_{t,1},\cdots, \hat{\bm{\omega}}_{t,N}]$. The QTZ process addresses dual objectives: converting the preference scores $\hat{\bm{\omega}}_{t,n}$ into valid integer degradation levels $a_{t,n} \in \mathcal{A}$, while generating diverse high-quality candidates that maintain real-time computational efficiency. Our solution integrates two coordinated operations: 1) quantizing the preference $\hat{\bm{\Omega}}_t$ into $K_t$ binary actions $\{\bm{\Omega}_t^{(k)}\}_{k=1}^{K_t}$ followed by transformation to $\mathcal{C}_t=\{\bm{a}_t^{(k)}\}_{k=1}^{K_t}$, and 2) dynamically adjusting $K_t$ to balance exploration and computational load. We now formalize these two operations sequentially.

\subsubsection{Preference-Guided Candidate Generation}
To translate preference scores into executable actions, we generate $\lfloor K_t/2 \rfloor$ candidates through direct interpretation of $\hat{\bm{\Omega}}_t$, reserving the remaining $K_t - \lfloor K_t/2 \rfloor$ actions for noise-enhanced exploration. 
\begin{itemize}
	\item Direct Interpretation: The first candidate $\bm{\Omega}_t^{(1)}=[\bm{\omega}_{t,n}^{(1)}, \cdots, \bm{\omega}_{t,N}^{(1)}]^\top$ is constructed deterministically by selecting the maximum preference score per device. For each binary vector $\bm{\omega}_{t,n}^{(1)}$, the element $\omega_{t,n,i}^{(1)} = 1$ at $i = \arg\max_{j\in\mathcal{A}} \hat{\omega}_{t,n,j}$, with other elements set to 0. As an example in Fig. \ref{MiLoRA_2}, when $\hat{\bm{\omega}}_{t,N}=[0.75, 0.92,\cdots,0.13]$ with $\hat{\omega}_{t,N,1} = 0.92$ as the maximum, we have $\bm{\omega}_{t,N}^{(1)} = [0,1,\cdots,0]$. Subsequent candidates $\{\bm{\Omega}_t^{(k)}\}_{k=2}^{\lfloor K_t/2 \rfloor}$ are generated through $\lfloor K_t/2 \rfloor-1$ independent probabilistic samplings. For each ED$_n$, we convert the preference scores $\hat{\bm{\omega}}_{t,n}$ to a categorical distribution $\hat{\bm{q}}_{t,n}=\{\hat{q}_{t,n,i}\}_{i=0}^{A-1}$ via softmax normalization:
	\begin{equation}
		\hat{q}_{t,n,i} = \frac{e^{\hat{\omega}_{t,n,i}}}{\sum_{i=0}^{A-1} e^{\hat{\omega}_{t,n,i}}}.
	\end{equation}
	We then sample the degradation level $i$ for each ED$_n$ from the categorical distribution $\hat{\bm{q}}_{t,n}$, setting $\omega_{t,n,i}^{(k)} = 1$ for the selected level and all other elements to 0.
	\begin{figure}
		\centering
		\includegraphics[scale=0.61]{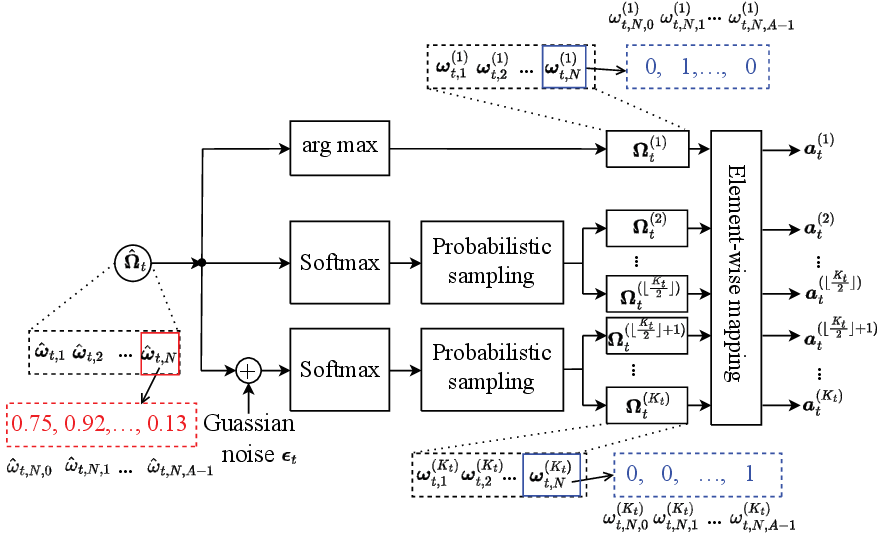}
		\captionsetup{font=footnotesize}
		\caption{The quantizer in the DNN-based actor module.}
		\label{MiLoRA_2}
	\end{figure}
	\item Noise-Enhanced Exploration: The remaining $K_t - \lfloor K_t/2 \rfloor$ candidates leverage Gaussian noise perturbations to explore novel action combinations:
	\begin{equation}
		\hat{\bm{\Omega}}_t^\epsilon = \hat{\bm{\Omega}}_t + \bm{\epsilon}_t, 
	\end{equation}
	where $\bm{\epsilon}_t\in\mathcal{N}(\bm{0}_{NA\times 1}, \bm{I}_{NA\times NA})$ denotes an $NA$-dimensional Gaussian distribution with zero mean and unit variance. Each candidate is generated by applying identical softmax normalization and categorical sampling as in direct interpretation, but to the noise-perturbed outputs $\hat{\bm{\Omega}}_t^{\epsilon}$.
\end{itemize}

The generated binary candidates $\{\bm{\Omega}_t^{(k)}\}_{k=1}^{K_t}$ are then transformed into executable degradation actions $\mathcal{C}_t = \{\bm{a}_t^{(k)}\}_{k=1}^{K_t}$ through element-wise mapping. For each candidate $k$ and ED $n$, the degradation level $a_{t,n}^{(k)}$ is determined by the active position in its one-hot representation $\bm{\omega}_{t,n}^{(k)}$:
\begin{equation}
	a_{t,n}^{(k)} = i \cdot \mathbb{I}(\omega_{t,n,i}^{(k)} = 1),
\end{equation}
where $i \in \mathcal{A} = \{0, \dots, A-1\}$. This yields candidate vectors $\bm{a}_t^{(k)} = [a_{t,1}^{(k)}, \dots, a_{t,N}^{(k)}]^\top$. As exemplified in Fig. \ref{MiLoRA_2} for EN$_N$, when $\bm{\omega}_{t,N}^{(1)} = [0,0,\dots,1]$, the transformation gives $a_{t,N}^{(1)} = 1$.
 
\begin{Remark}
	The preference-guided candidate generation balances exploitation and exploration via a dual-stream architecture, enabled by one-hot encoding. It expands the integer action space into a higher-dimensional binary domain, enhancing candidate diversity while enforcing action feasibility \cite{okadaEfficient2019}. Within this expanded space, the direct interpretation stream exploits DNN knowledge to sample top-preference actions that concentrates on high-value solutions. Simultaneously, The noise-enhanced stream discovers novel actions via controlled perturbations, focusing exploration near high-confidence predictions. The coordinated interaction between expanded representation and dual-stream mechanisms enables efficient solution exploration and exploitation.
\end{Remark}

\subsubsection{Adaptive Candidate Set Adjustment}
The candidate set size $K_t$ critically influences the efficacy of LAB framework. While larger candidate sets theoretically enable more comprehensive exploration of the action space, they substantially increase the evaluation time of acquisition function. Furthermore, expansion of $K_t$ does not guarantee monotonic performance improvement since the acquisition function serves as a surrogate optimizer rather than directly maximizing the true utility $U_t$. This calls for an adaptive mechanism that dynamically balances exploration breadth against computational burden.

In this paper, we implement an adaptive $K_t$ mechanism updated every $\Delta_{\rm K}>0$ time slots. We begin by quantifying the alignment between candidate actions and the DNN preference through the $L_2$-distance metric:
\begin{equation}
	v_t^{(k)} = \|\bm{\Omega}_t^{(k)} - \hat{\bm{\Omega}}_t\|_2.
\end{equation}
The candidates are sorted in ascending order of $v_t^{(k)}$, yielding an ordered sequence $\{\bm{\Omega}_t^{(1)}, \bm{\Omega}_t^{(2)}, \dots, \bm{\Omega}_t^{(K_t)}\}$, where $\bm{\Omega}_t^{(1)}$ denotes the action closest to the DNN preference. The optimal candidate $\bm{a}_t^\ast$ is then determined by solving \eqref{prob_acqmax} over this ordered set. Let $k_t^\ast \in \{1, \dots, K_t\}$ represent the position index of $\bm{a}_t^\ast$ within the sorted sequence. For time slot $t$, we update $K_t$ as
\begin{subnumcases}{K_t=\!\label{m0_update}}
	K_1, \!\!\!\!& $t=1$,\\
	\min\!\left(\!\max\!\left(\!\left\{k^\ast_{i}\right\}_{i=t-1}^{\Delta_{\rm K}}\!\right)\!\!+\!\!1,K_1\right), \!\!\!\!& $t \!\!\!\!\!\mod\! \Delta_{\rm K} \!\!=\!\! 0$,\\
	K_{t-1},\!\!\!\! &\text{otherwise},
\end{subnumcases}
where $K_1$ is initialized as a positive integer at time slot $t=1$. Intuitively, $K_t$ is primarily determined by $k_t^\ast$. A higher $k_t^\ast$ indicates that optimal actions reside in less explored regions farther from current preferences, suggesting the need for candidate set expansion; a lower $k_t^\ast$ suggests high-utility actions are clustered near existing preferences, enabling potential set reduction.

\section{Policy Update Module} \label{sec5}
The policy update module facilitates continuous performance enhancement in LAB by dynamically optimizing the parameters for both the actor (DNN) and critic (BO) modules. Leveraging accumulated historical execution data, it periodically retrains the DNN-based action generator and re-calibrates the GP model underpinning the BO acquisition function. The following subsections detail the parameter optimization procedures for both the actor and critic modules.

\subsection{Parameters update of DNN}
We implement an experience replay mechanism to update the DNN policy. Upon completion of time slot $t$, the optimal state-action pair $(\bm{s}_t, \bm{a}_t^\ast)$ is stored in a finite-capacity replay memory of size $J_{\rm D}>0$. Once the memory reaches capacity, older samples are replaced by newer entries. The training of DNN commences when the memory contains at least $J_{\rm D}/2$ samples and recurs periodically every $\Delta_{\rm D}>0$ slots. When $t \mod \Delta_{\rm D}=0$, we uniformly sample a minibatch $\{(\bm{s}_\nu, \bm{\Omega}_{\nu}^\ast), \nu\in \mathcal{T}_t^{\rm D}\}$ from the memory, where $\bm{\Omega}_{\nu}^\ast$ is the binary representation of $\bm{a}_{\nu}^\ast$ and $\mathcal{T}_t^{\rm D}$ denotes selected temporal indexes, respectively. The DNN parameter $\bm{\theta}_{\rm dnn}$ is updated by minimizing the binary cross-entropy loss via the Adam optimizer:
\begin{equation}
	\begin{split}
		Loss(\bm{\theta}_{\rm dnn}) = -\frac{1}{J_{\rm D}^{\rm s}}\sum_{\nu\in \mathcal{T}_{\rm D}}\left[\left(\bm{\Omega}_{\nu}^\ast\right)^{\top}\log f_{\bm{\theta}_{\rm dnn}}\left(\bm{s}_\nu\right) \right. \\
		\left.+\left(1-\bm{\Omega}_{\nu}^\ast\right)^{\top}\log\left(1-f_{\bm{\theta}_{\rm dnn}}\left(\bm{s}_t\right)\right)\right],
	\end{split}
\end{equation}
where $\log$ denotes the element-wise logarithm, and $J_{\rm D}^{\rm s} = |\mathcal{T}_t^{\rm D}|$ is the minibatch size. Following parameter update, the actor module utilizes the new $\bm{\theta}_{\rm dnn}$ to generate candidate actions. This procedure enables the DNN to continuously assimilate high-quality state-action mappings from recent experiences, progressively improving decision effectiveness.

\subsection{Parameters update of BO}
The BO parameter $\bm{\theta}_{\rm{bo}}=\{\bm{\theta}_{\rm h}, \upsilon_{\rm a}, \rho, \delta_\epsilon\}$ is updated through a dedicated observation cache of size $J_{\rm B}>0$. At the end of time slot $t$, the optimal degradation action $\bm{a}_t^\ast$, channel state $\bm{h}_t$, and temporal index $t$ collectively form the augmented input $\bm{z}_t^\ast = (\bm{h}_t, \bm{a}_t^\ast, t)$. While the solution to the degradation action selection problem \eqref{prob3} yields an exact value $U_t^\ast$ for the input $\bm{z}_t^\ast$, this value represents a noisy observation $y_t^\ast$ of the utility function due to the unknown system dynamics. Consequently, we set $y_t^\ast \leftarrow U_t^\ast$ and store the tuple $(\bm{z}_t^\ast, y_t^\ast)$ in the BO cache, with older entries evicted upon overflow. Periodically every $\Delta_{B}$ time slots, i.e., when $t \mod \Delta_{B}=0$, the parameter $\bm{\theta}_{\rm bo}$ is re-optimized by maximizing the log marginal likelihood over the cached dataset $\mathcal{D}_t^{\rm B}=\{(\bm{z}_i^\ast, U_i^\ast)\}_{i=t-J_{\rm B+1}}^{t}$ \cite{brochuTutorial2010}:
\begin{equation}\label{prob_boupdate}
	\bm{\theta}_{\rm bo} \leftarrow \arg\max_{\bm{\theta}} \log q\left(\bm{Y}_t \mid \bm{Z}_t\right),
\end{equation}
where $q\left(\bm{Y}_t \mid \bm{Z}_t\right)\sim \mathcal{N}(\bm{0},\bm{K}_t+\delta_\epsilon^2\bm{I})$ integrates the GP prior \eqref{eq_gp_prior} and likelihood \eqref{eq_likelihood}, and has the following closed-from log expression:
\begin{equation}\label{eq_marglikeli}
	\begin{split}
	\log q\left(\bm{Y}_t \mid \bm{Z}_t\right)=&-\frac{1}{2}\left(\bm{Y}_t\right)^\top\left(\bm{K}_t + \delta_\epsilon^2\bm{I}_{J_{\rm B}\times J_{\rm B}}\right)^{-1}\bm{Y}_t \\
	&- \frac{1}{2}\log|\bm{K}_t + \delta_\epsilon^2\bm{I}_{J_{\rm B}\times J_{\rm B}}|.
	\end{split}
\end{equation}
By maximizing \eqref{eq_marglikeli} via gradient-based method (e.g., quasi-Newton method), we obtain the optimized BO parameter $\bm{\theta}_{\rm bo}$. This updated parameter is then used in the acquisition function during the next time slot $t+1$ to select the optimal degradation action.

\begin{Remark}
	LAB fundamentally redefines the actor-critic paradigm through a hybrid architecture that integrates DRL with BO. Unlike conventional actor-critic frameworks where both components utilize DNNs, LAB employs a BO critic to guide a DNN-based actor. Specifically, the BO critic leverages GP posterior with estimated mean and uncertainty to derive high-value actions, thereby enabling uncertainty-aware explorations which is unavailable in conventional DNN critics. Concurrently, the DNN actor learns from these high-value actions to progressively improve its action generation, and thus achieving data-efficient policy update.
\end{Remark}

\section{Simulation Results}\label{sec6}
%
%\begin{table}
%	\centering
%	\caption{Simulation Parameters}
%	\label{Paraset}
%	\begin{tabular}{|c|c|c|}
%		\hline
%		$p = 20$ dBm& $e_k^{\rm pre}=50$ mW & $\tau_k^{\rm pre}=0.05$ s \\
%		\hline
%		$B = 4$ MHz & $k_{\rm c}=1.2\times10^{-30}$  & $C = 10$ cycles/operation \\
%		\hline
%		$\delta^2 = -174$ dBm/Hz  & $F = 64 \times 2.8$ GHz & $f_{\rm max} = 2.4$ GHz\\
%		\hline
%	\end{tabular}
%\end{table}

In this section, we evaluate the system performance through numerical simulations, considering an EI system comprising one ES and $N=3$ EDs by default. The ES employs YOLOv11x \cite{jocherUltralytics2025} for object detection tasks, processing images from the UrbanRoad Self-Driving dataset \cite{Self2025}. This dataset contains $M=15,000$ RGB video frames at native resolution  $(\iota_n^{\rm w},\iota_n^{\rm h})=(1920,1200)$, where each frame represents an input image. At the time slot $t=0$, each ED$_n$ initializes with a uniform random starting frame index $m_{0,n} \sim \mathcal{U}[0, M-1]$, then sequentially advances through the dataset such that $m_{t,n} = (m_{0,n} + t) \mod M$ each time slot. For each image, ED$_n$ applies Guassian pyramid downsampling across $A=4$ degradation levels. The IoU threshold is set as $\gamma_{\rm th}=0.5$.

To model autonomous driving scenarios, ED mobility follows a 100 m $\times$ 50 m rectangular trajectory centered at the ES location. The system operates over $T=3,000$ discrete time slots. At $t=0$, initial positions of EDs are randomly distributed along a designated 100-meter edge of the rectangle. Between consecutive time slots, ED positions progress in discrete 2.5-meter increments counterclockwise along the rectangular trajectory. We adopt Rayleigh fading channels for wireless communications between EDs and ES. The channel gain for ED$_n$ at time slot $t$ is modeled as $h_{t,n} = \varsigma \cdot \bar{h}_{t,n}$, where $\varsigma$ follows an exponential distribution with unit mean. The average channel gain $\bar{h}_{t,n}$ is defined as $\bar{h}_{t,n} = G_A \left( \frac{3 \times 10^8}{4 \pi f_{\rm c} d_{t,n}} \right)^{\lambda}$, where $G_A=4.11$ denotes antenna gain, $f_{\rm c}=2.4$ GHz is the carrier frequency, $d_{t,n}$ represents the Euclidean distance between ED$_n$ and ES at time $t$, and $\lambda=2.4$ by default is the path-loss exponent. We consider homogeneous EDs with identical transmit power $p_n=0.1$ watt, and weighting factors $w_n=w=1$ by default, $\forall n$. The total bandwidth is $W=5$ MHz. The noise power spectrum density at the ES is $\delta^2=-174$ dBm/Hz. {\color{black}{In practice, the degradation latency $\tau_{t,n}^d$ in \eqref{eq_latencyd} and the edge computation latency $\tau_{t,n}^c$ in \eqref{eq_latencyc} are hard to model precisely with closed-form functions under dynamic hardware conditions. Instead, we measure these latencies in real time during execution.}}

For the proposed framework LAB, the actor module employs a Transformer-based DNN consisting of three core components: a fully-connected input embedding layer that generates 64-dimensional embeddings, a 2-layer Transformer encoder with 4 attention heads per layer and 256-dimensional feed-forward networks, and a fully-connected output layer using sigmoid activation to produce $\hat{\bm{\Omega}}_t$. The default parameters of DNN include: input state history length $l=1$, replay memory size $J_{\rm D}=512$, minibatch size $J_{\rm D}^{\rm s}=128$, initial number of candidates $K^0=\min\{8N, A^N\}$ with update interval $\Delta_{\rm K}=32$, DNN training interval $\Delta_{\rm D}=20$, and the Adam optimizer learning rate $\xi=0.01$. The critic module employs the L-BFGS-B method to update the BO parameters, with defaults including observation cache size $J_{\rm B}=256$ and BO update interval $\Delta_{B}=20$. {\color{black}{The exploration coefficient in the UCB acquisition function \eqref{eq_acqfunc} is $\zeta=\sqrt{0.2}$. Remarkably, even with this non-optimized $\zeta$ setting, our proposed LAB framework consistently outperforms both the standalone BO and state-of-the-art (SOTA) DRL baselines, as demonstrated later in Sec \ref{SecVIIB}.}} All the simulations are conducted on an Intel(R) Xeon(R) Silver 4110 (2.10 GHz) machine with a Tesla P40 GPU. 

\subsection{Parameter Sensitivity Analysis of LAB}

\begin{figure}
	\centering
	\includegraphics[scale=0.5]{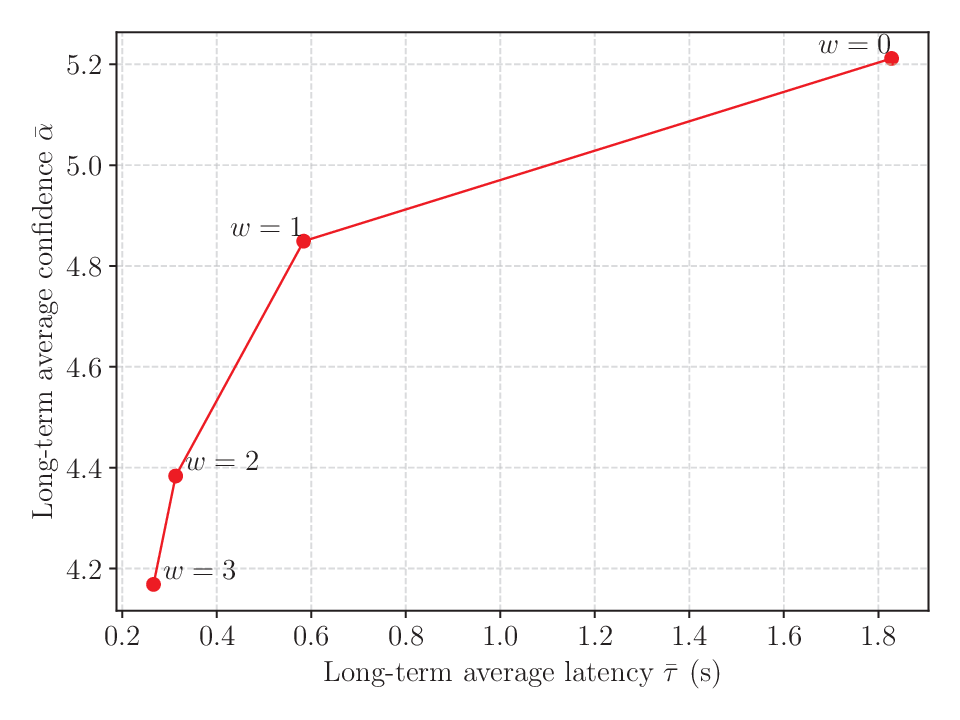}
	\captionsetup{font=footnotesize}
	\caption{The confidence-latency trade-off of LAB. }
	\label{sim_conf_vs_latency_w}
\end{figure}

In Fig. \ref{sim_conf_vs_latency_w}, we demonstrate the confidence-latency trade-off of LAB under different delay weighting factors. Specifically, we compute the long-term average confidence $\bar{\alpha}=\frac{1}{T}\sum_{t=0}^{T-1}\sum_{n=1}^{N}\alpha_{t,n}$ and the long-term average latency $\bar{\tau}=\frac{1}{T}\sum_{t=0}^{T-1}\sum_{n=1}^{N}\tau_{t,n}$. {\color{black}{Note that a higher $\bar{\alpha}$ generally indicates more or more confident detections.}} As $w$ increases from $0$ to $3$, LAB increasingly prioritizes latency reduction, resulting in lower $\bar{\tau}$ at the expense of reduced $\bar{\alpha}$. Notably, increasing $w$ from 0 to 1 reduces $\bar{\tau}$ dramatically by 68.3\% (from 1.83 s to 0.58 s) while only decreasing $\bar{\alpha}$ by 6.9\% (from 5.21 to 4.85). Therefore, we select $w=1$ for all subsequent simulations as it maintains high confidence without incurring substantial latency penalties.

\begin{figure}
	\centering
	\includegraphics[scale=0.45]{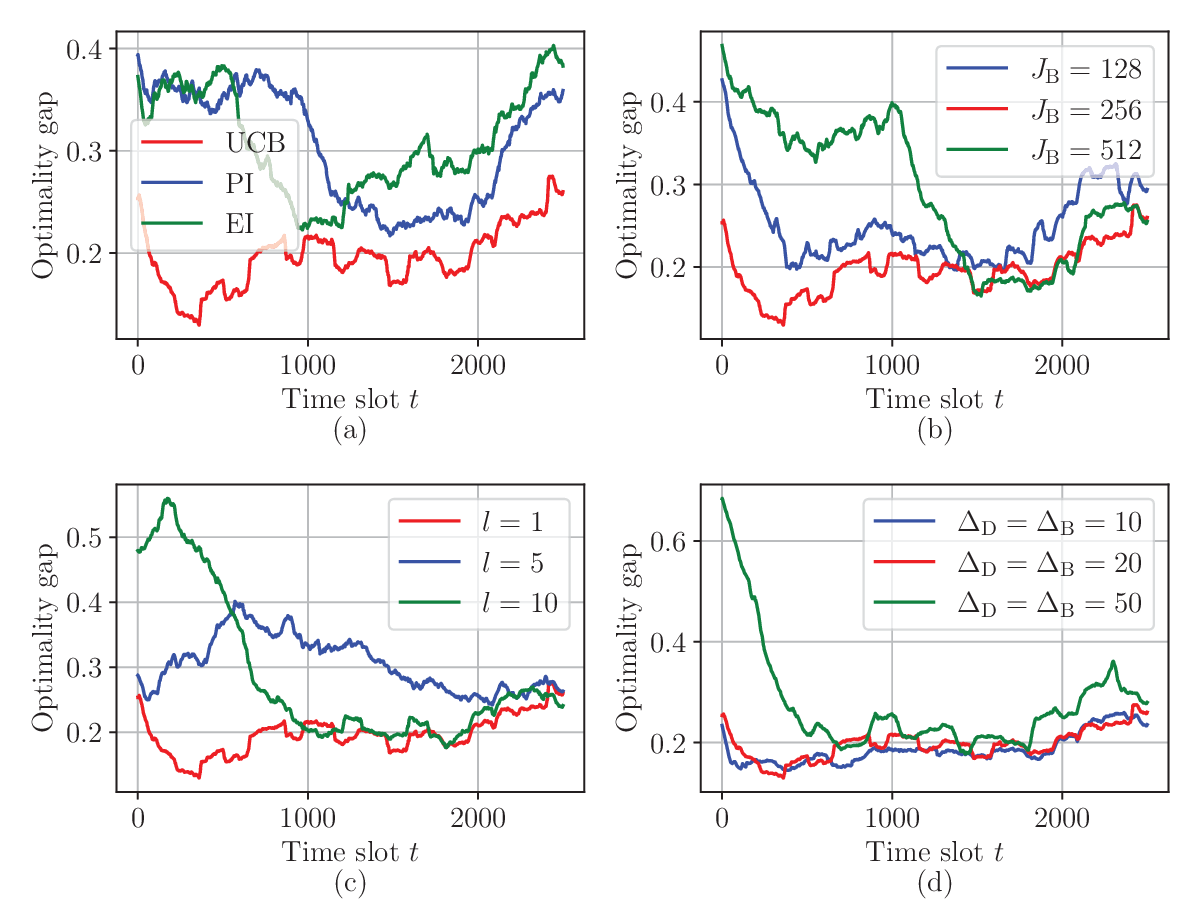}
	\captionsetup{font=footnotesize}
	\caption{Optimality gap under different LAB parameters: (a) acqusition function; (b) BO cache size; (c) history length of input state; and (d) update interval.}
	\label{sim_optgap_vs_t_under_paras}
\end{figure}

In Fig. \ref{sim_optgap_vs_t_under_paras}, we analyze parameter sensitivity of LAB through the optimality gap metric $U_t^{\ast}-U_t$, where $U_t^{\ast}$ represents the theoretical maximum utility obtained through exhaustive search over all $A^N$ degradation actions for \eqref{prob3}, with results window-averaged over 500 time slots. Fig. \ref{sim_optgap_vs_t_under_paras}(a) reveals that UCB acquisition function consistently outperforms expected improvement (EI) and probability of improvement (PI) alternatives by more effectively balancing exploration-exploitation trade-offs. {\color{black}{Note that the fluctuations in the optimality gap, such as the noticeable rise around $t=2000$, stems from the non-stationary nature of the real-world dataset}}. For the BO cache in Fig. \ref{sim_optgap_vs_t_under_paras}(b), a size of $J_{\rm B}=256$ achieves the minimal optimality gap, while a smaller cache ($J_{\rm B}=128$) exhibits degraded performance due to insufficient historical data utilization, and a larger cache ($J_{\rm B}=512$) introduces noisy temporal correlations that compromise decision quality. State history configuration in Fig. 4(c) reveals that $l=1$ (current state only) delivers optimal performance, whereas extended histories (e.g., $l=10$) provides only marginal asymptotic improvement at increased computational overhead. Besides, we notice from Fig. 4(d) that while reducing $\Delta_{\rm D}$ and $\Delta_{\rm B}$ generally lowers optimality gaps, the improvement plateaus below 20 time slots. Further reduction to $\Delta_{\rm D}=\Delta_{\rm B}=10$ provides only marginal gains while substantially increasing computational load. To balance system performance with implementation efficiency, we adopt UCB acquisition, $J_{\rm B}=256$, $l=1$, and $\Delta_{\rm D}=\Delta_{\rm B}=20$ as the default configuration.

{\color{black}
\subsection{Ablation Study}
\begin{figure}
	\centering
	\includegraphics[scale=0.5]{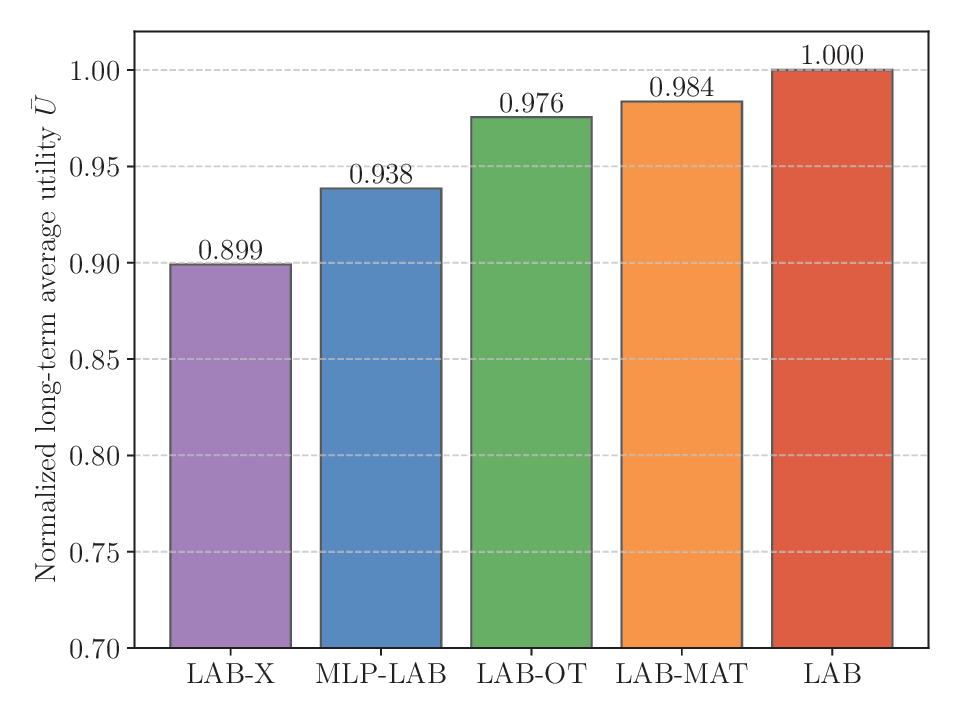}
	\captionsetup{font=footnotesize}
	\caption{{\color{black}Normalized long‑term average utility of LAB and its ablated variants: MAT ($\mathrm{Mat\acute{e}rn}$ kernel), OT (no temporal kernel), MLP (MLP-based actor), and X (BO with contextual input).}}
	\label{sim_ablation}
\end{figure}

{To validate the design choices of our LAB framework, we conduct an ablation study, with results presented in Fig. \ref{sim_ablation}. We compare the long-term average utility of our default LAB against four ablated variants: LAB‑MAT (replaces the RBF kernel with a $\mathrm{Mat\acute{e}rn}$ kernel), LAB‑OT (removes the temporal kernel), MLP‑LAB (substitutes the Transformer-based actor with a multi‑layer perceptron), and LAB‑X (augments the kernel with image-content features). The results are normalized to that of LAB. We observe utility drops of 1.64\% for LAB-MAT, 2.45\% for LAB-OT, and 6.26\% for MLP-LAB, respectively. These results confirm the effectiveness of our kernel design and actor architecture. Notably, LAB-X exhibits a more pronounced degradation of 10.09\%. This counter-intuitive result likely stems from the dynamic and non‑stationary image content in the real‑world dataset, where rapidly changing scenes make historical content information misleading rather than informative.}}

\subsection{System Performance Comparison} \label{SecVIIB}
To validate the efficacy of LAB, we evaluate it against five representative benchmarks:
\begin{itemize}
	\item IDEAL: Selects degradation actions by exhaustive search over all $A^N$ possibilities per slot to directly maximize utility $U_t$. While providing the theoretical performance upper bound for problem \eqref{prob2}, its requirement of completing full action-space evaluation before decision-making makes it computationally infeasible for practical implementation.
	\item Standard BO (BO) \cite{galanopoulosAutoML2021a}: Generates degradation actions by maximizing the acquisition function \eqref{eq_acqfunc} per slot through full enumeration of all $A^N$ possibilities.
	\item DBAG \cite{wangEdge2023}: A SOTA DRL method that formulates problem \eqref{prob2} as an MDP and applies soft actor-critic for degradation action optimization.
	\item Delay-oblivious offloading (DelayObli): A detection-optimal strategy where all EDs persistently select the lowest degradation level $a_{t,n}=0$, transmitting original full-resolution images without degradation.
	\item Delay-minimizing offloading (DelayMin): An offloading-latency-optimal approach where all EDs consistently select the highest degradation level $a_{t,n}=A-1$, transmitting minimally-sized images to the ES.
\end{itemize}
{\color{black}{For a fair comparison, all training-based methods (BO, DBAG, and our proposed LAB) are evaluated under an online training paradigm without using any pre-trained models. Besides, we strictly adhere to the original neural network architecture of DBAG \cite{wangEdge2023} for its actor and critic networks.}} Following discrete degradation action selection, all benchmark methods derive optimal bandwidth allocations by solving subproblem \eqref{prob_bandalloc}. While commercial CVX solvers such as MOSEK \cite{apsMOSEK2025} guarantee optimal solutions for \eqref{prob_bandalloc}, they incur substantial computational overhead when deployed on an ES equipped with low-frequency CPU. This limitation is especially pronounced for the IDEAL method, which requires solving \eqref{prob_bandalloc} for $A^N$ times per slot due to exhaustive search. {\color{black}{To mitigate computational complexity, we implement a DNN-based bandwidth allocation solver following \cite{wangEdge2023}, which efficiently approximates solutions to \eqref{prob_bandalloc}.}} Fig. \ref{sim_ideal_cvx_vs_dnn} compares the performance of MOSEK and our DNN solver when implementing the IDEAL method over 1,000 time slots. The results demonstrate that the DNN solver achieves almost identical per-slot utility to the optimal solution. Based on this validation, we employ the DNN solver for all subsequent method evaluations.  We conduct Monte Carlo simulations across all methods using 10 random seeds. Each simulation spans $T=3,000$ time slots, with final results averaged over all seeds.

%offers 15.5\%, 21.2\%, 9.2\%, 3.2\% lower WSEC on average than ECO, LCO, RS, and 2-Stage, respectively, and achieves a near-optimal performance with average 1.3\% optimality gap to CD. 

\begin{figure}
	\centering
	\includegraphics[scale=0.5]{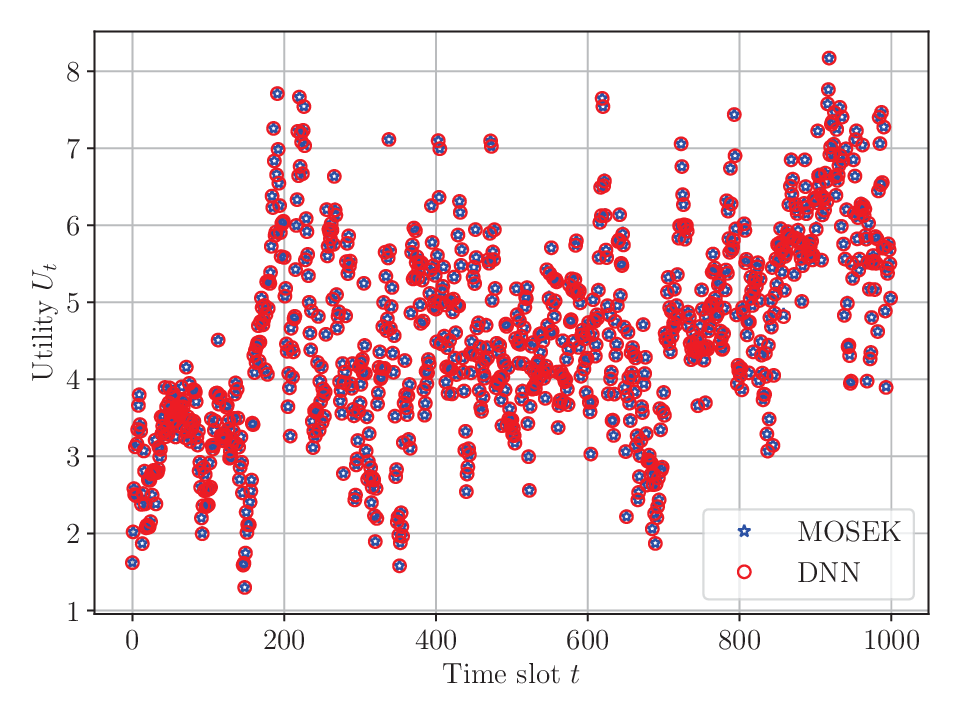}
	\captionsetup{font=footnotesize}
	\caption{Per-slot utility of DNN-based solver for bandwidth allocation.}
	\label{sim_ideal_cvx_vs_dnn}
\end{figure}

\begin{figure}
	\centering
	\includegraphics[scale=0.45]{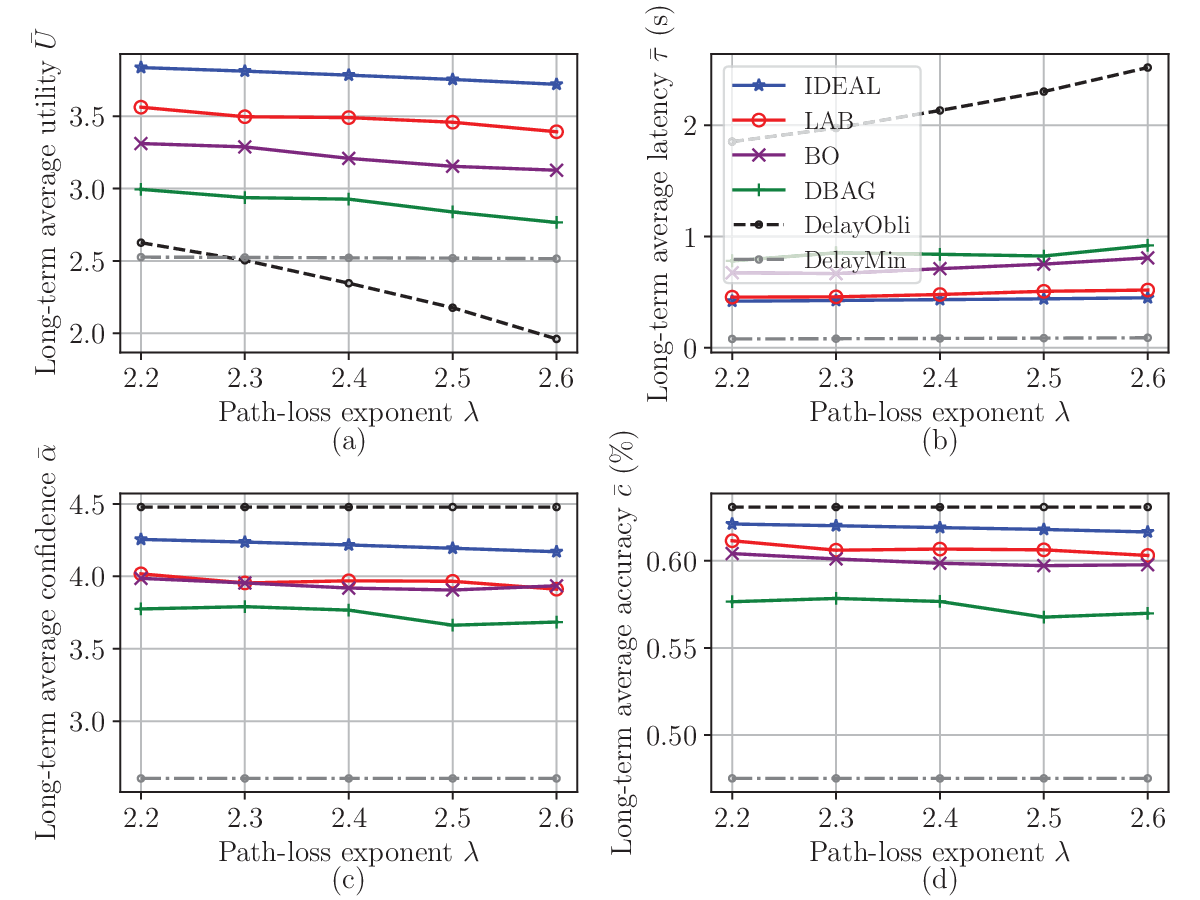}
	\captionsetup{font=footnotesize}
	\caption{The system performance under various path-loss exponent $\lambda$.}
	\label{Sim_perform_vs_pathloss}
\end{figure}

We first evaluate the system performance under varying path-loss exponent $\lambda$ in Fig. \ref{Sim_perform_vs_pathloss}. As shown in Fig. \ref{Sim_perform_vs_pathloss}(a), the long-term average utility $\bar{U} = \frac{1}{T}\sum_{t=0}^{T-1}U_t$ decreases monotonically with $\lambda$ across all methods due to deteriorating channel conditions. Non-adaptive strategies (DelayObliv and DelayMin) yield the lowest utilities as they fail to balance detection quality against transmission latency. LAB achieves an average of 7.98\% optimality gap of utility relative to IDEAL, outperforming BO and DBAG by 6.95\% and 15.55\% in gap reduction, respectively. Fig. \ref{Sim_perform_vs_pathloss}(b) presents the long-term average latency $\bar{\tau}$. As expected, DelayMin maintains minimal latency through persistent highest-degradation transmissions, while DelayObli incurs maximum latency due to full-resolution image offloading without any degradation. LAB exhibits marginally higher latency than IDEAL (within 0.07 s across all $\lambda$ values), yet reduces latency by 32.96\% and 42.60\% on average compared to BO and DBAG. Figs. \ref{Sim_perform_vs_pathloss}(c) and (d) simultaneously depict long-term average confidence $\bar{\alpha}$ and accuracy $\bar{c} = \frac{1}{T}\sum_{t=0}^{T-1}\sum_{n=1}^{N}c_{t,n}$, both exhibiting similar trends across $\lambda$ variations. Focusing on practical detection performance, LAB achieves near-ideal accuracy with only 1.22\% average deficit relative to IDEAL, while surpassing BO and DBAG by 0.7\% and 3.29\%, respectively. Collectively, these results demonstrate the ability of LAB to provide enhanced detection accuracy while simultaneously reducing latency compared to both BO and DBAG.

\begin{figure}
	\centering
	\includegraphics[scale=0.45]{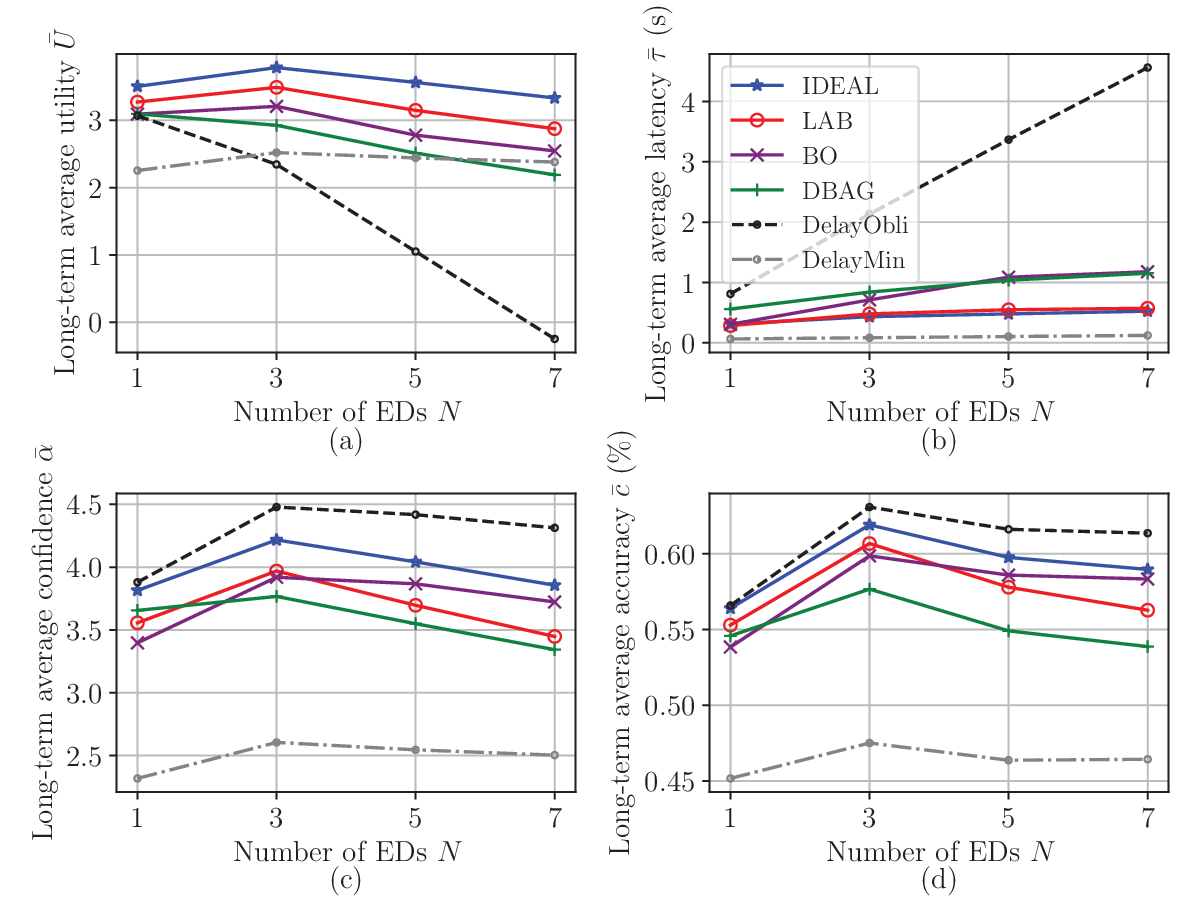}
	\captionsetup{font=footnotesize}
	\caption{The system performance under various number of EDs $N$.}
	\label{Sim_perform_vs_N}
\end{figure}

In Fig. \ref{Sim_perform_vs_N}, we examine the system scalability under varying numbers of EDs $N$. In Fig. \ref{Sim_perform_vs_N}(a), we observe a notable crossover between DelayObliv and DelayMin at $N=3$. This occurs because fixed total bandwidth allocation reduces available bandwidth per ED as $N$ increases, causing DelayObliv's full-resolution transmissions to incur prohibitive latency. Simultaneously, DBAG exhibits diminishing utility with network expansion, ultimately performing below DelayMin at $N=7$. In contrast, LAB achieves an average optimality utility gap of 9.94\% relative to IDEAL, improving upon BO and DBAG by 8.21\% and 14.58\% on average, respectively. Fig. \ref{Sim_perform_vs_N}(b) reveals escalating latency across all methods with rising $N$. LAB maintains latency comparable to IDEAL (0.04 s higher on average), while achieving substantial latency reductions relative to BO (35.08\% lower) and DBAG (47.35\% lower) at $N=7$. Figs. \ref{Sim_perform_vs_N}(c) and (d) demonstrate similar confidence and accuracy trends under network scaling. Among all the considered methods, DelayObliv maintains the highest accuracy by consistently transmitting full-resolution images, while DelayMin yields the lowest accuracy through persistent high-level degradation. LAB maintains an average of 1.75\% accuracy below IDEAL. Besides, it matches BO's accuracy within 0.15\% and surpasses DBAG by 2.25\% on average. These results verify the ability of LAB to scale across different network sizes.

\begin{figure}
	\centering
	\includegraphics[scale=0.45]{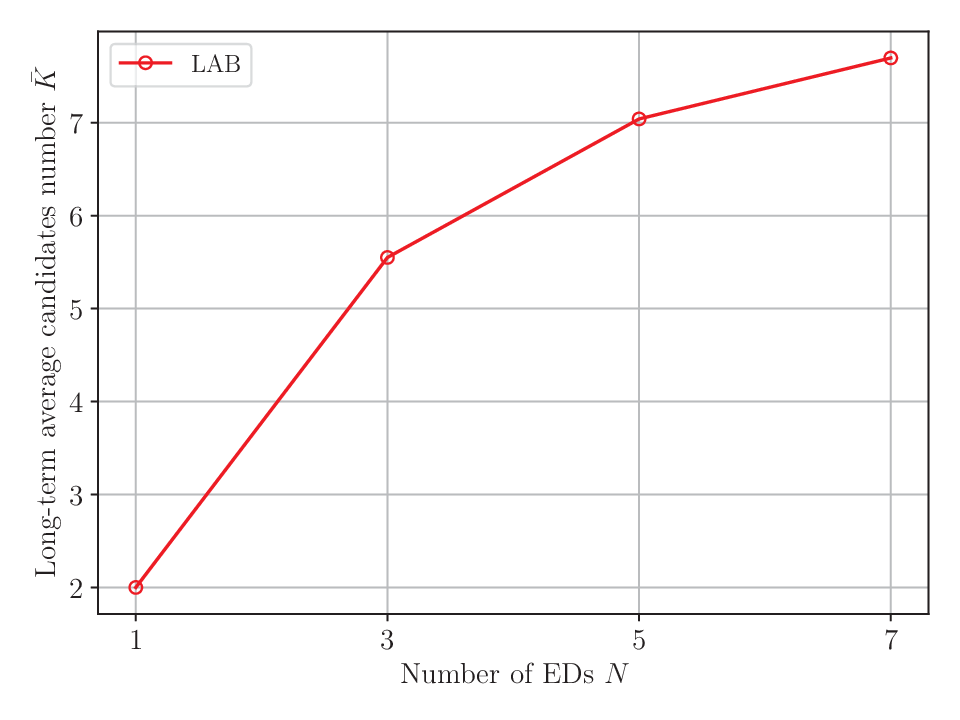}
	\captionsetup{font=footnotesize}
	\caption{The average number of candidates under various number of EDs $K$.}
	\label{Sim_candi_vs_N}
\end{figure}

To quantify the computational efficiency of LAB, Fig. \ref{Sim_candi_vs_N} illustrates the long-term average number of candidate actions $\bar{K}=\frac{1}{T}\sum_{t=0}^{T-1} K_t$ across varying numbers of EDs $N$. The result shows that $\bar{K}$ increases almost linearly with $N$. Unlike BO which exhaustively evaluates all $A^N$ actions to maximize the acquisition function, LAB drastically reduces computational overhead especially at large $N$. For example, when $N=7$ and $A=4$, BO requires $4^7=16,384$ evaluations per slot. In contrast, LAB maintains $\bar{K}\approx7.7$, reducing per-slot evaluations by 3 orders of magnitude (over 2,000 times). This near-linear scaling of $\bar{K}$ with $N$ underscores the suitability of LAB for large-scale networks where standard BO becomes computationally prohibitive.

%\begin{table}
%	\centering
%	\caption{Computation time (s)}
%	\label{Table_CompTime}
%	\begin{tabular}{|c|c|c|c|c|c|}
%		\hline
%		Number of SDs &4 & 6& 8& 10& 12\\
%		\hline
%		CD & 5.52 & 13.01 & 23.04 & 35.66& 50.71 \\
%		\hline
%		LOP & 0.08 & 0.11 & 0.14& 0.17 & 0.20\\
%		\hline
%	\end{tabular}
%\end{table}

\section{Conclusion and Discussions}
This paper investigated task-oriented computation offloading in a multi-device EI system for visual object detection, aiming to maximize long-term accuracy of all EDs while minimizing end-to-end inference latency under fixed bandwidth constraints. We formulated the problem as a black-box MINLP, where the challenges lie in the content-dependent accuracy, absence of analytical models, and combinatorial complexity. To solve this problem, we proposed LAB, a novel learning framework that seamlessly integrates DRL and BO. Specifically, LAB employed a DNN-based actor to generate degradation actions and a BO-based critic with Gaussian process surrogates for optimal action selection, complemented by convex optimization for bandwidth allocation. Extensive evaluations on a real-world self-driving dataset demonstrated that LAB achieves near-ideal accuracy-latency trade-offs. Crucially, it outperforms conventional DRL in both accuracy and latency performance, while surpassing standard BO with equivalent accuracy yet significantly faster inference and substantially reduced computational overhead.

\begin{figure}
	\centering
	\includegraphics[scale=0.5]{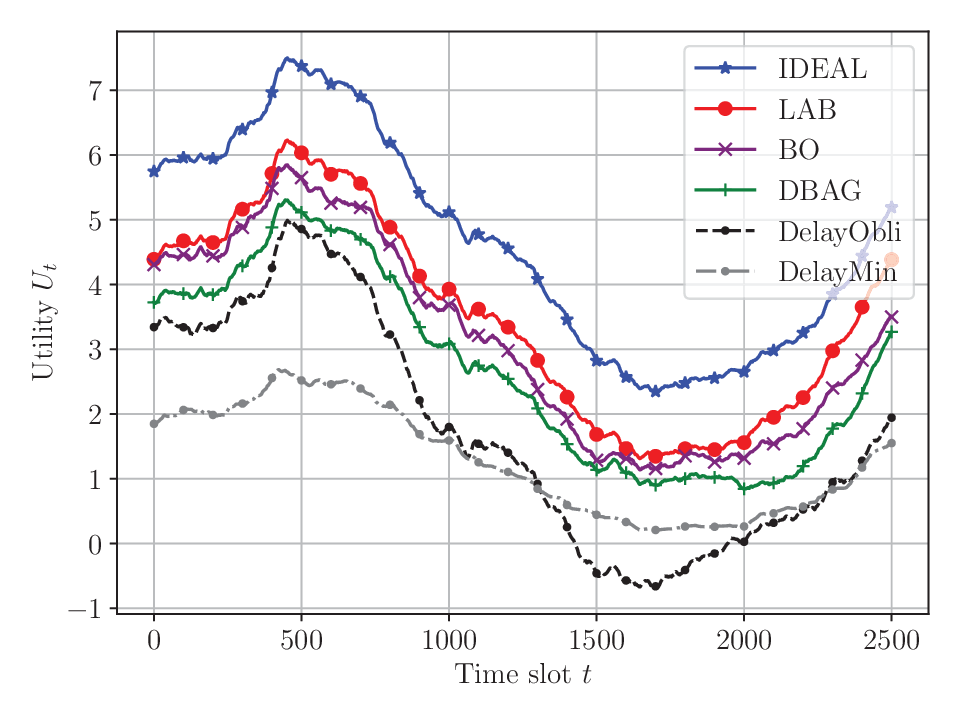}
	\captionsetup{font=footnotesize}
	\caption{{\color{black}{Performance comparison of LAB and baseline methods on VisDrone2019 dataset.}}}
	\label{sim_UAVdataset_utility}
\end{figure}

{\color{black}{We conclude this paper with an extended evaluation of our LAB framework on the real-world VisDrone2019 dataset \cite{zhuDetection2022}. In this simulation, we set $\omega = 4$ and $\zeta = \sqrt{2}$. Fig. \ref{sim_UAVdataset_utility} plots the moving average of the utility $U_t$ over the last 500 slots. The results show that LAB consistently outperforms all baseline methods in this new scenario, improving the utility in average by 11.33\%, 27.98\%, 90.91\%, and 167.30\%, respectively. This outcome supports the potential generalizability of LAB and encourage its further exploration in diverse edge‑AI settings. Notably, the high scene complexity and object density of VisDrone2019 results in a 23.95\% performance gap between LAB and the IDEAL upper bound. This highlights a clear direction for future work, where further gains could be achieved through specialized optimizations in kernel design, DNN architecture, and hyper‑parameters.}}

\bibliographystyle{IEEEtran}
\bibliography{IEEEabrv,MyRefLib}

%\end{IEEEbiography}
%\begin{IEEEbiography}[{\includegraphics[width=0.95in,height=1.4in]{SuzhiBi.jpg}}]{Suzhi Bi (S'10-M'14-SM'19)}
%received his Ph.D. degree in information engineering from The Chinese University of Hong Kong in 2013. From 2013 to 2015, he was a post-doctoral research fellow with the Department of Electrical and Computer Engineering, National University of Singapore. Since 2015, he has been with the College of Electronics and Information Engineering, Shenzhen University, China, where he is now a Professor. His research interests mainly involve in optimization and machine learning techniques for wireless information and power transfer, mobile computing, and wireless sensing. He received 2019 IEEE ComSoc Asia-Pacific Outstanding Young Researcher Award, 2021 IEEE ComSoc Asia-Pacific Outstanding Paper Award, and conference Best Paper Awards of IEEE SmartGridComm 2013, IEEE/CIC ICCC 2021, and IEEE VTC-Spring 2022. He is an Editor of IEEE Transactions on Wireless Communications and IEEE Wireless Communications Letters.
%\end{IEEEbiography}
\end{document}